\documentclass[times,sort&compress,3p]{elsarticle}
\journal{}
\usepackage[labelfont=bf]{caption}

\usepackage{amsmath,amsfonts,amssymb,amsthm,booktabs,color,epsfig,graphicx,
url,float,enumitem,bm}

\theoremstyle{plain}
\newtheorem{theorem}{Theorem}

\newtheorem{assumption}{Assumption}[section]

\newtheorem{lemma}[theorem]{Lemma}
\newtheorem{corollary}[theorem]{Corollary}

\theoremstyle{definition}

\newtheorem{remark}{Remark}

\def\ad{{\boldsymbol{A}_{D}}}
\def\ab{\boldsymbol{a}}
\def\hataab{{\hat{\boldsymbol{a}}}}
\def\hatAb{{\widehat{\boldsymbol{A}}}}

\def\Ab{{\boldsymbol{A}}}
\def\Bb{{\boldsymbol{B}}}
\def\bd{{\boldsymbol{b}}} 
\def\Db{{\boldsymbol{D}}}
\def\dd{{\mathrm{d}}}
\def\dg{{d_g}}
\def\hb{{\boldsymbol{H}}}

\def\bigi{{\boldsymbol{\mathrm{I}}}}
\def\Qb{{\boldsymbol{Q}}}

\def\hatqb{{\widehat{\boldsymbol{Q}}}}
\def\qtil{{\widetilde q}}
\def\ub{{\boldsymbol{u}}}
\def\Ub{{\boldsymbol{U}}}
\def\ugj{{u_{jg}}}

\def\Ugj{{U_{jg}}}
\def\vb{{\boldsymbol{v}}}
\def\Vb{{\boldsymbol{V}}}
\def\vo{{v_0}}
\def\Vo{{V_0}}
\def\vg{{v_g}}
\def\Vg{{V_g}}
\def\vtil{{\widetilde v}}
\def\Vtil{{\widetilde V}}

\def\what{{\widehat w}}
\def\wtil{{\widetilde w}}
\def\util{{\widetilde u}}
\def\Wtil{{\widetilde W}}
\def\wtilb{{\widetilde{\boldsymbol{w}}}}
\def\wb{{\boldsymbol{w}}}
\def\Wb{{\boldsymbol{W}}}
\def\zb{{\boldsymbol{z}}}
\def\Zb{{\boldsymbol{Z}}}
\def\ob{{\boldsymbol{0}}}

\def\gammab{{\boldsymbol{\Gamma}}}

\def\btheta{\boldsymbol{\theta}}
\def\eps{{\epsilon}}
\def\epsbar{{\bar\epsilon}}
\def\epsb{{\boldsymbol{\epsilon}}}
\def\sigb{{\boldsymbol{\Sigma}}}
\def\Psib{{\boldsymbol{\Psi}}}

\def\p{\partial}
\def\e{{\mathrm{E}}}
\def\cov{{\mathrm{Cov}}}
\def\diag{\mathop{\rm diag}}

\def\Var{\mathop{\rm Var}}
\def\Cor{\mathop{\rm Cor}}

\def\intt{\int_{0}^{1}}

\begin{document}

\begin{frontmatter}

\title{High-dimensional factor copula models with estimation of latent variables}

\author[1]{Xinyao Fan\corref{mycorrespondingauthor}}
\author[2]{Harry Joe}

\address{Department of Statistics, University of British Columbia, Vancouver, Canada V6T 1Z4}

\cortext[mycorrespondingauthor]{Corresponding author. Email address: \url{xinyao.fan@stat.ubc.ca}}

\begin{abstract}
Factor models are a parsimonious way to explain the dependence
of variables using several latent variables. 
In Gaussian 1-factor and structural factor models (such as bi-factor, oblique factor)
and their factor copula counterparts, 
factor scores or proxies are defined as
conditional expectations of latent variables given the observed variables.
With mild assumptions, the proxies
are consistent for corresponding latent variables as the sample size and the number of
observed variables linked to each latent variable go to infinity. 
When the bivariate copulas linking observed variables to latent variables
are not assumed in advance, sequential procedures are used 
for latent variables estimation, copula family selection and parameter
estimation. {The use of proxy variables for factor copulas means that approximate
log-likelihoods can be used to estimate copula parameters with less
computational effort for numerical integration.}
\end{abstract}

\begin{keyword} 
bi-factor\sep
{factor scores}\sep
identifiability\sep
oblique factor \sep
proxy variables \sep
tail dependence 
\end{keyword}

\end{frontmatter}

\section{Introduction\label{sec:1}}

Factor models are flexible and parsimonious ways to explain the
dependence of variables with one or more latent variables. The
general factor copula models in \cite{krupskii2013factor} and
\cite{krupskii2015structured} are extensions of classical Gaussian
factor models and are useful for joint tail inference if the variables
have stronger tail dependence that can be expected with Gaussian models,
such as with asset return data.

In classical factor analysis, estimates of
the latent variables referred to factor scores 
(see \cite{johnson2002applied} and \cite{gorsuch1983factor})
are of interest and useful for interpretation and further analysis. 
For similar reasons, the inference of the latent variables in factor copulas is also useful. 

For maximum likelihood estimation in parametric factor copula models,
the copula density and likelihood involve integrals with dimension
equal to the number of latent variables.
\cite{krupskii2013factor} and \cite{krupskii2015structured} provide
procedures for computationally efficient evaluations of the log-likelihood
and its gradient and Hessian for 1-factor, 2-factor, bi-factor
and a special case of the oblique factor copula models.
These are the cases for which integrals can be evaluated via
1-dimensional or 2-dimensional Gaussian-Legendre quadrature. 
Bi-factor and oblique factor models are useful when the observed {variables}
can be placed into several non-overlapping homogeneous groups.

In this paper, one main focus for factor copulas is to show
how use of ``proxies" to estimate latent variables (a) can help in diagnostic
steps for deciding on the bivariate copula families that link observed
variables to the latent variables and (b) lead to approximate
log-likelihoods for which numerical maximum likelihood estimation 
is much faster.
The 1-factor, bi-factor and oblique factor copula models are 
used to illustrate the theory because with their previous numerical
implementations for maximum likelihood, we can make comparisons with the
faster proxy-based methods introduced within. The theory developed here
can be applied in other factor copula models, and this is discussed in
the final section on further research.

\cite{krupskii2022approximate} initiate the use of proxies for latent
variables to speed up numerical maximum likelihood estimation;
their approach involved unweighted means in 1-factor and
unweighted group means for oblique factor copula models.
Their approach does not extend to bi-factor and other structural $p$-factor
copula models.
In order to accommodate these other factor copula models, we
use two-stage proxies, with stage 1 being factor scores based on the estimated loading matrix
after each variable has been empirically transformed to standard normal,
and stage 2 based on conditional expectations of latent variables
given the observed variables (using a copula model fitted from
the stage 1 proxies).

To justify the sequential method for latent variable and copula model
estimation, several theoretical results in consistency, as the number $D$ of
observed variables increase to infinity, are needed. 
Factor copula models would mainly be considered 
if the observed variables are monotonically
related and have at least moderate dependence. There may be more
dependence in the joint tails than expected with Gaussian dependence, but
the Gaussian factor models can be considered as first-order models.

The proxies as estimates of latent variables are 
extensions of Gaussian regression factor scores because these are based on
conditional expectations of latent variables given the observed.
For $D$ increasing, we consider the observed variables
(or their correlations, partial correlations,
or linking copulas) as being sampled from a super-population.
We first obtain conditions for the proxies or conditional
expectations to be asymptotically
consistent estimates (of corresponding latent variables) when the
the factor model is completely known; we also have results that
suggest rates of convergence.
In cases where consistency is not possible, then we know that
we cannot expect consistency when parameters in the factor model must
be estimated. One such case involves the Gaussian bi-factor model
where (a) a loading matrix of less than full column rank implies that
the latent variables are not identifiable, and the model can be reduced
to an oblique factor model;  and (b) the rate of convergence of the
proxies is slow if the loading matrix has a large condition number.
With a sample of size $N$,
the assumption of
a super-population, combined with factor models being closed
under margins, suggest that in the case of estimated parameters, 
(i) all parameters can be estimated with $\sqrt{N}$ consistency
and (ii) proxies are consistent under mild conditions.
Because we need a method of proof that is valid for both Gaussian
and factor copula models, our technique is different from that
of \cite{bai2012statistical}. Their approach does not provide
insights for non-identifiability of latent variables such as
in the bi-factor model.

The remainder of the paper is organized as follows. 
Section \ref{sec-factor} provides the representations of Gaussian factor
models and (structured) factor copula models.
Section \ref{sec-proxies} has expressions for conditional expectations,
and the proxies as estimates of the latent variables.
Section \ref{sec-consistency1} has some sufficient or necessary
conditions for asymptotic consistency of proxy variables 
with known loading matrix or known linking copulas.
Section \ref{sec-consistency2} has results and conditions
for  the consistency of proxy variables with
estimated parameters in linking copulas (with copula families known).
Section \ref{sec-sequential} proposes a sequential method for the practical use of
proxy variables in cases where the linking copula families are not
specified. 
Simulation studies in Section \ref{sec-simulation} show the proxies are useful
in selecting linking copula families and getting accurate parameter
estimates with less computing time. 
Section \ref{sec-residdep} has sufficient conditions for using the
proxies in the Section \ref{sec-proxies} when 
observed variables have weak dependence, rather than independence,
conditional on the observed variables. 
Section \ref{sec-discussion} has a summary and discussion for further research.

\section{Structured factor copula models}
\label{sec-factor}

The 1-factor, bi-factor and oblique factor models that are the main focus
of this paper are shown graphically in Figure
\ref{fig:1fact_graph},
\ref{fig:bi_graph} and \ref{fig:obl_graph} respectively.
The graphical representations are valid for the Gaussian factor models
and their extensions to factor copula models.

\begin{figure}[H]
\centering
\includegraphics[height=3.5cm]{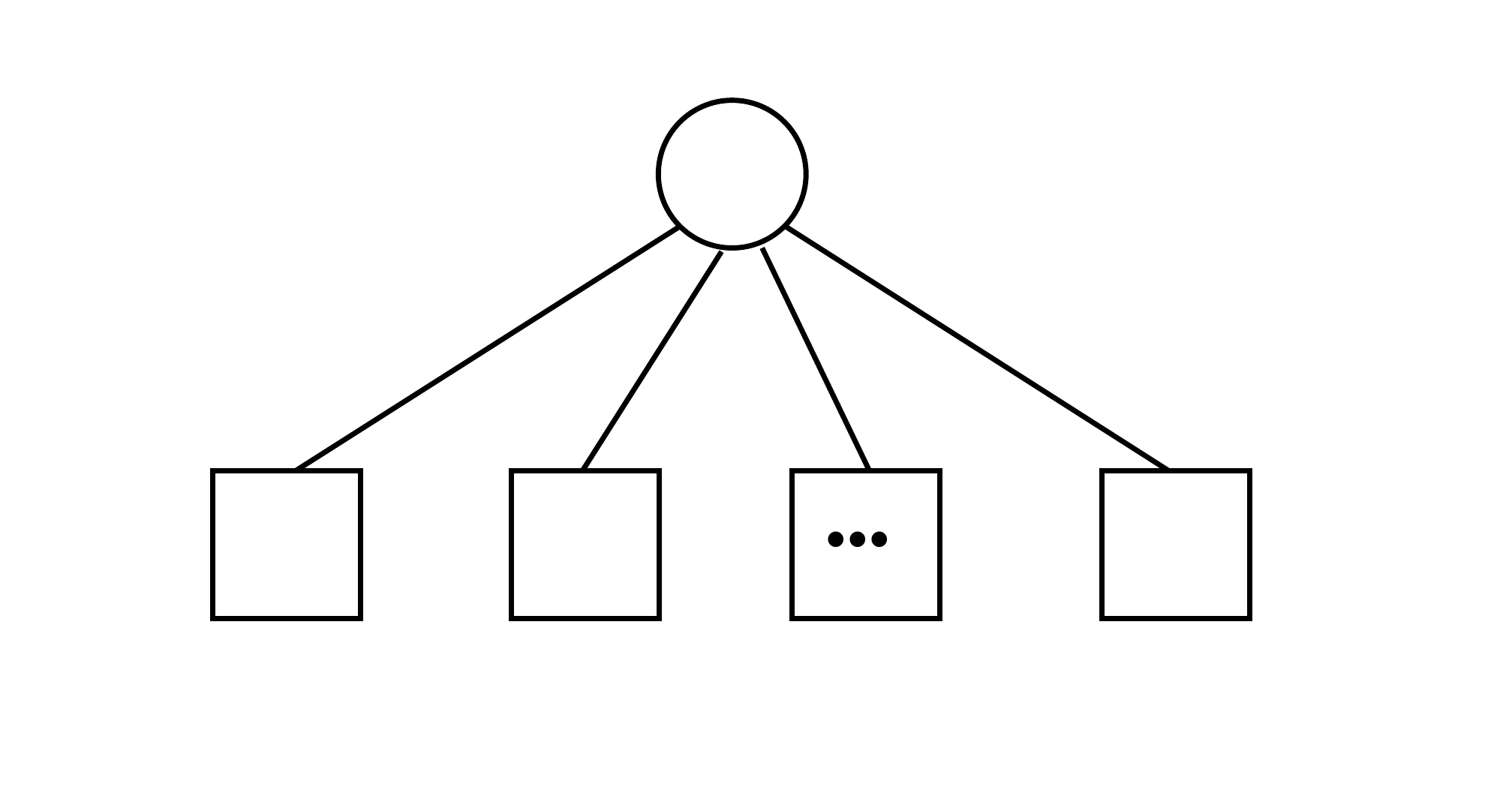}
\caption{1-factor model, a sequence of the observed
variables (rectangular shapes) are linked to the latent variable (circles).}
\label{fig:1fact_graph}
\end{figure}

\begin{figure}[H]
\centering
\includegraphics[height=3.5cm]{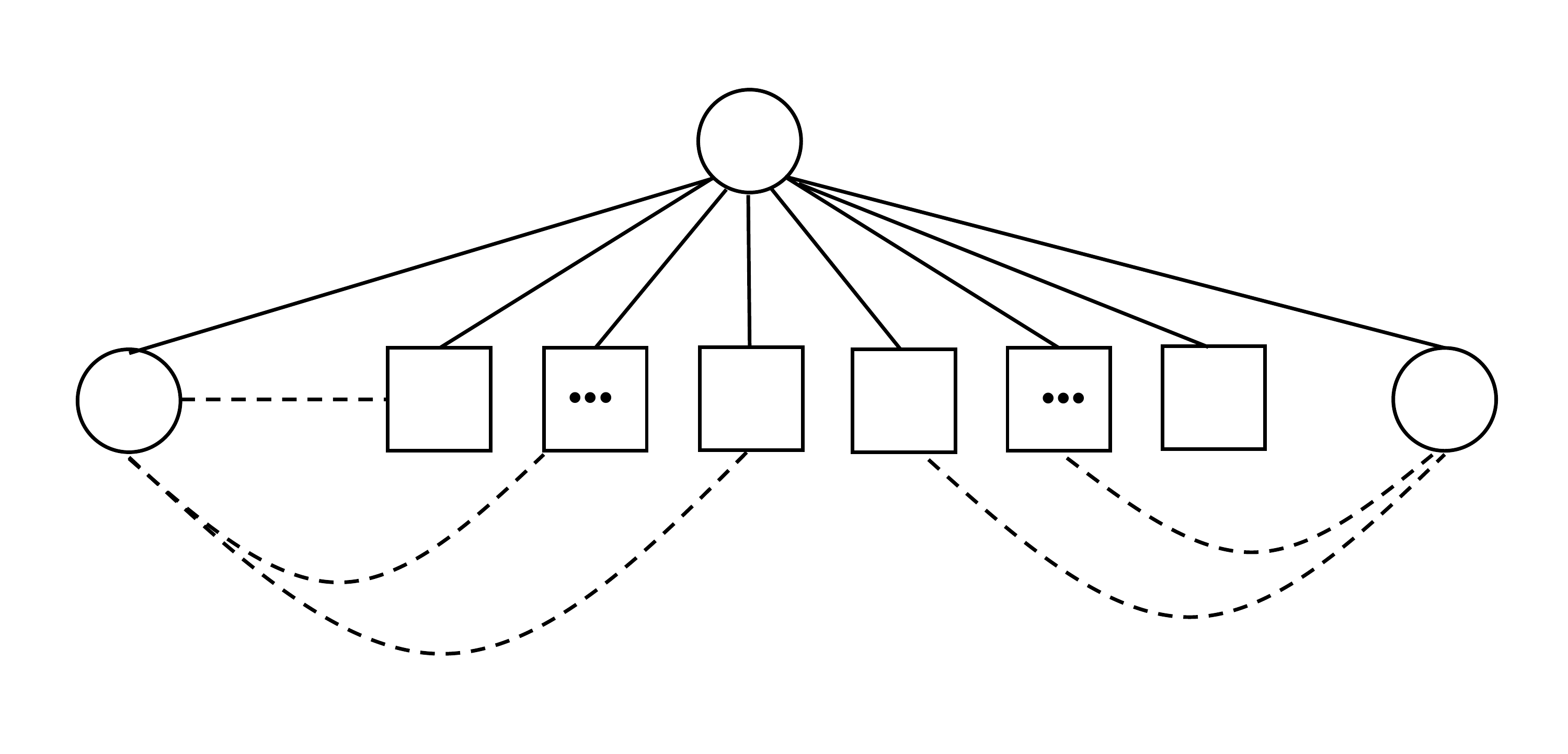}
\caption{Bi-factor model, $G$ local latent
variables (circles) and a sequence of observed variables (rectangles)
are all linked to the global latent variable (circle;
root variable) in solid lines; the dashed lines indicate the conditional
dependence between observed variables and local latent variables
conditioned on the global latent variable. 
Illustration here has $G=2$ groups.} 
\label{fig:bi_graph}
\end{figure}

\begin{figure}[H]
\centering
\includegraphics[height=4cm]{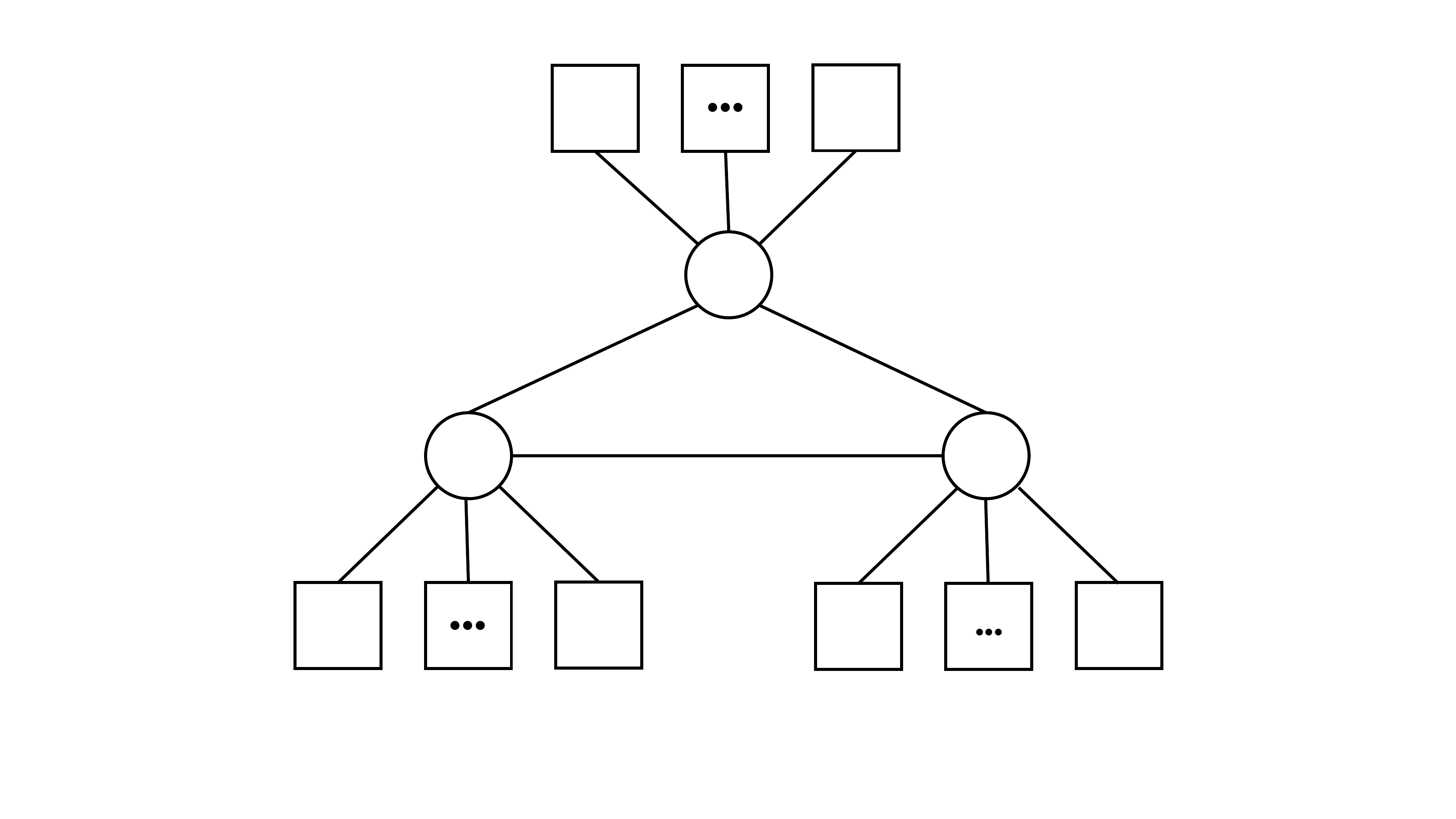}
\caption{Oblique factor model, observed variables (rectangles)
are divided into the several non-overlapping groups, latent
variables (circles) are dependent and the observed variables 
in each group are linked to the (local) group latent variable.
Illustration here has $G=3$ groups.} 
\label{fig:obl_graph}
\end{figure}

$p$-factor models with $p\ge2$ are in general do not have an identifiable
loading matrix in the Gaussian case because of orthogonal transform
of the loading matrix. The bi-factor structure is a special of the
$p$-factor model with many structural zeros.
The bi-factor and oblique factor models are two parsimonious factor models
that can be considered when variables can be divided into $G$
non-overlapping groups.

In the graphs for 1-factor and bi-factor, each observed variable links
to the (global) latent variable;
the edges of the graphs have a correlation (of observed with latent)
for multivariate Gaussian and bivariate linking copula for the factor copula.

For the bi-factor graph, there are additional edges linking each
observed variable to its (local) group latent variable.
For multivariate Gaussian, these edges have partial correlation 
of observed variables with corresponding group latent variable,
conditioned on the global latent variable; this can be converted to
a linear representation with a loading matrix --- see Section 6.16 of
\cite{Joe2014}. For the bi-factor copula, these additional edges
are summarized with bivariate copulas linking observed variables with the
corresponding group latent variable, conditioned on the global latent variable.
The group latent variables are independent of each other and are independent of the global latent variable. There is dependence
of all variables from the common link to the global latent variable. There is additional within-group dependence from links to the
group latent variable.

For the oblique factor graph, each observed variable is linked to a (local) group latent variable, so that there is within-group
dependence. The group latent variables are dependent, and these leads to
between-group dependence. 

The linear representations (Gaussian) and copula densities are
given below, with notation in a form that allows for their study
as the number of observed variables $D$ increases to $\infty$.
References are their derivations are in
Section 3.10 and 3.11 of \cite{Joe2014},
\cite{krupskii2013factor}, \cite{krupskii2015structured}
and \cite{krupskii2022approximate}.

For notation, observed variables are denoted as $U_{j}$ or $U_{jg}$
after transform to $U(0,1)$, or  $Z_{j}$ or $Z_{jg}$ after transform to $N(0,1)$, and
latent variables are denotes as $V$, $V_0$ or $V_g$ on the $U(0,1)$ scale
and $W$, $W_0$ or $W_g$ on the $N(0,1)$ scale. Copula densities for different vectors are indicated using $C$ with
subscripts for random vectors. The generic notation for a bivariate 
copula cdf has the form $C_{U,V}(u,v)$ and
its partial derivatives are denoted as $C_{U|V}(u|v)=\p C_{U,V}(u,v)/\p v$
and $C_{V|U}(v|u)=\p C_{U,V}(u,v)/\p u$ because these are conditional 
distributions.
Lower case variables are used as
arguments of densities or dummy variables of integrals.

\textbf{1-factor copula model} with $D$ variables $\Ub_D =(U_1,\ldots,U_D)$,
the copula density is:
\begin{equation}
  c_{\Ub_{D}}(\ub_{D})= \intt \prod_{j=1}^{D}c_{jV}(u_j,v)\,\dd v,
  \label{eq-1factor-copula-pdf}
\end{equation}  
where $c_{jV}=c_{U_j,V}$ for all $j$.

\textbf{Bi-factor copula model}  with $\dg$ variables
in group $g$, {fixed $G$ groups}, $\Db=(d_1,\ldots,d_G)$ and
$D=\sum_{g=1}^G \dg$ total number of variables.
The copula density is:
\begin{equation}
  c_{\Ub_{\Db}}(\ub_{\Db}) =  \intt 
  \prod_{g=1}^{G}\Bigl\{\intt
\prod_{j=1}^{\dg}c_{\Ugj \Vo}(\ugj,\vo)\cdot
   c_{\Ugj\Vg;\Vo}\bigl(C_{\Ugj|\Vo}(\ugj|\vo),\vg \bigr) \dd v_g
  \Bigr\} \dd v_0.
  \label{eq-bifactor-copula-pdf}
\end{equation}
The notation $c_{\Ugj\Vg;\Vo}$ is the copula density assigned to the edge
with connecting $\Ugj, \Vg$ given $\Vo$.

\textbf{Oblique factor copula model}  with $\dg$ variables
in group $g$ as above. The copula density is:
\begin{equation}
  c_{\Ub_{\Db}}(\ub_{\Db}) =  \intt \cdots \intt \prod_{g=1}^{G}
  \prod_{j=1}^{\dg}c_{\Ugj,\Vg}(\ugj,\vg)  \, c_{\Vb}(\vb)\, 
  \dd v_1\cdots \dd v_G.
  \label{eq-oblique-copula-pdf}
\end{equation}  
The notation $c_{\Vb}(v)$ is the joint copula density of the latent variables.

When all linking copulas are bivariate Gaussian copulas, the usual representation
of Gaussian factor models result after transforms of $U(0,1)$ variables
to standard normal $N(0,1)$ variables.

\textbf{Gaussian 1-factor} 
\begin{equation}
    Z_{j}=\alpha_{j}W+\psi_{j}\epsilon_{j}\quad 
  j\in\{1,\ldots,D\},
  \label{eq-1factor-gauss}
\end{equation}  
where $W,\epsilon_1,\epsilon_2,\ldots$ are mutually independent $N(0,1)$
random variables, and $-1<\alpha_j<1$ and $\psi_j^2=1-\alpha_j^2$
for all $j$.

\textbf{Gaussian bi-factor} 
\begin{equation}
    Z_{jg}=\alpha_{jg,0}W_{0}+\alpha_{jg}W_{g}+\psi_{jg}\epsilon_{jg}\quad 
  j\in \{1,\ldots,\dg\};\quad g\in\{1,\ldots,G\},
  \label{eq-bifactor-gauss}
\end{equation}  
where $W_0,\{W_g\},\{\epsilon_{jg}\}$ are mutually independent $N(0,1)$,
$\alpha$'s are in $(-1,1)$ and
$\psi_{jg}^2=1-\alpha_{jg,0}^2-\alpha_{jg}^2<1$ for all $(j,g)$.
Note that 
$\rho_{Z_{jg},W_{g};W_{0}}=\alpha_{jg}/(1-\alpha_{jg,0}^2)^{1/2}$
is the partial correlation of $Z_{jg}$ with $W_{g}$ given $W_{0}$.

\textbf{Gaussian oblique factor} 
\begin{equation}
  Z_{jg}=\alpha_{jg}W_{g}+\psi_{jg}\epsilon_{jg}\quad j\in\{1,\ldots,\dg\}; 
  \quad g\in\{1,\ldots,G\},
  \label{eq-oblique-gauss}
\end{equation}  
where $\{\epsilon_{jg}\}$ are mutually independent $N(0,1)$, independent of the multivariate normal vector $(W_1,\ldots,W_G)$,
with zero mean vector and unit variances,
$\alpha$'s are in $(-1,1)$ and
$\psi_{jg}^2=1-\alpha_{jg}^2$ for all $(j,g)$.
Let $\bm{\Sigma}_{W}$ be the correlation matrix of $\Wb=(W_1,\ldots,W_G)^T$.

\textbf{Matrix representation}
\begin{equation} 
\label{eq-factor_mat}
 \Zb_{D}=\Ab_{D}\Wb+\Psib_{D}\epsb_{D},
\end{equation} 
where the loading matrix $\Ab_{D}$ is of size $D\times p$,
$\Psib_{D}^{2}$ is a $D \times D$ diagonal matrix of individual variances
($\psi_j$ or $\psi_{jg}$),
$\epsb_{D}$ is a $D\times1$ column vector of $\eps_j$ or $\eps_{jg}$, and $\Zb_{D}$ is a $D\times1$ column vector of $Z_j$ or $Z_{jg}$.
For 1-factor, $p=1$; for bi-factor, $p=G+1$ and $\Wb=(W_0,W_1,\ldots,W_G)^T$, 
and for 
oblique factor, $p=G$ and $\Wb=(W_1,\ldots,W_G)^T$.

A matrix identity that is useful in calculations of conditional
expectation and covariance of $\Wb$ given $\Zb_D$ is the following
\begin{equation} 
\label{eq-matrix-id}
  \Ab^T_D(\Ab_D\Ab^T_D+\Psib_D^2)^{-1} = 
  (\bigi_p+\Ab_D^T\Psib_D^{-2}\Ab_D)^{-1}\Ab_D^T\Psib_D^{-2},
\end{equation} 
when $\Ab_D\Ab_D^T+\Psib_D^2$ is non-singular and $\Psib_D^2$ has no
zeros on the diagonal.
This identity is given in \cite{johnson2002applied}.

\medskip

In general, except for the case of bivariate Gaussian linking copulas,
the integrals in the above copula densities do not simplify,
and numerical maximum likelihood involves numerical integration when there is a random sample of size $N$.
With a parametric family for each bivariate linking copula, \cite{krupskii2013factor} and \cite{krupskii2015structured}
outline numerically efficient approaches for a modified Newton-Raphson
method for optimizing the negative log-likelihood for 1-factor, bi-factor
and a special nested factor subcase of the oblique factor copulas. 1-dimensional or 2-dimensional
Gauss-Legendre quadrature is used to evaluate the integrals and their
partial derivatives order 1 and 2 for the gradient and Hessian
of the negative log-likelihood. The factor copulas are extensions of their Gaussian counterparts and are useful when
plots of normal scores data (after empirical transforms to $N(0,1)$) show
tail dependence or tail asymmetry in bivariate scatterplots.

In classical factor analysis (\cite{johnson2002applied}),
factor scores or estimates of latent variables are considered after estimating a loading matrix. The information on the latent variables may be used in subsequent analysis following the factor analysis; for example, regression analysis incorporated the factor information. For factor copulas, these could be obtained after fitting a parametric
model. Factor copula models for practical use are considered when Gaussian factor models are considered as first-order approximations, so that (transformed) factor scores could be
considered as a starting point. An approach to estimate the latent variables without the need to fit a parametric model by numerical procedures in \cite{krupskii2013factor} and \cite{krupskii2015structured} is proposed, and more details will be illustrated in later sections. 

The next section has proxies as estimates of latent variables
based on conditional expectations given observed variables.

\section{Proxies for the latent variables}
\label{sec-proxies}

In the Gaussian factor models, factor scores are defined as the estimates
of unobserved latent variables, see \cite{johnson2002applied}.
The form of factor scores that extend to factor copula models are the
{\it regression factor scores}, which are conditional expectations of
latent variables given the observed variables.

For factor copula models, having reasonable estimates of latent variables is also of interest since these can lead to
simpler and more efficient numerical procedures for determining parametric bivariate linking copula families and estimating their parameters. For factor copula models, we use the term `proxies' for the estimates
of latent variables, as in \cite{krupskii2022approximate}.

The study of the conditional expectation of latent variables
($W$'s or $V$'s) given observed variables ($Z$'s or $U$'s)
is done in three stages for the models in Section \ref{sec-factor}.

A. The loading matrix is known or all of the bivariate linking
copulas are known. In this case, the proxy variables are defined as the
conditional expectation of latent variables given the observed variables, and we refer to ``conditional expectation" proxies. 
 
B1. Gaussian factor models with estimated loading matrix. Since the general Gaussian factor model is non-identifiable in terms
of rotation of the loading matrix, for a model with two or more factors,
consistency of estimation requires a structured loading matrix such as
that of the bi-factor model or oblique factor model.
In the models, the proxies are defined in the same way as in case A but with
an estimated loading matrix (in blocks), where parameter estimates have
a variance of order $O(1/N)$ for sample size $N$.
 
B2. Factor copula models with known parametric families for each linking
copula. In the models, the parameters are estimated via sequential
maximum likelihood with a variance of order $O(1/N)$ for sample size
$N$. The proxies are defined in the same way as in case A with the estimated
linking copulas.

C. Linking copula families are not known or specified in advance
(the situation in practice). A sequential method is used starting with
unweighted averages as estimates in  \cite{krupskii2022approximate} or
regression factor scores computed from an estimated
loading matrix after observed variables are transformed to have $N(0,1)$
margins. Then, the ``conditional expectation" proxies are constructed
with the copula families and estimated parameters determined in the
first stage.

In Sections  \ref{sec-factorscore}
and \ref{sec-proxies-copula}, 
the conditional expectations (for case A) are given.
The asymptotic properties of proxies for case A and for cases B1, B2 are given in 
Section \ref{sec-consistency1} and \ref{sec-consistency2} respectively, and 
the sequential method of case C is in Section \ref{sec-sequential}.

\subsection{Proxies in Gaussian factor models}
\label{sec-factorscore}

In this section, we summarize $\e(\Wb|\Zb_D=\zb_D)$
for 1-factor, bi-factor and oblique factor Gaussian models with 
observed variables that are in $N(0,1)$.
These are called (regression) factor scores in the factor analysis
literature.

\noindent
\textbf{$p$-factor}: 
Let $(\wb^0,\zb_D)$ be a realization of $(\Wb,\Zb_D)$.
The proxy for $\Wb$ (or estimate of $\wb^0$) given $\zb_D$ are:
\begin{equation}
   \wtilb_{\Db}=\e(\Wb|\Zb_{D}=\zb_D)
  ={\Ab_{\Db}}^{T}(\Ab_{\Db}\Ab_{\Db}^{T}+\Psib_{\Db}^{2})^{-1}\zb_{\Db} 
  =
  (\bigi_{p}+{\Ab_{\Db}}^{T}\Psib_{\Db}^{-2}{\Ab_{\Db}})^{-1}{\Ab}_{\Db}^{T}
  {\Psib}_{\Db}^{-2}\zb_{\Db} ,
  \label{eq-1factor-fs}
\end{equation}
if $(\Ab_{\Db}\Ab_{\Db}^{T}+\Psib_{\Db}^{2})$
is non-singular
and ${\Psib_{\Db}}$ has no zeros on diagonal.
The above matrix equality follows from \eqref{eq-matrix-id}.

If $\Psib_{\Db}$ has zeros on diagonal, then a linear combination of the
latent variables is an observed variable.
If $(\Ab_{\Db}\Ab_{\Db}^{T}+\Psib_{\Db}^{2})$
is singular, then a linear combination of observed variables is 
a constant. 
These unrealistic cases will not be considered.

\noindent
\textbf{Bi-factor}:
Let $(w^0_0,w^0_1,\ldots,w^0_G,\zb_D)$ be a realization of 
$(W_0,W_1,\ldots,W_G,\Zb_D)$.
The proxies (or estimates of $w^0,w_1^0,\ldots,w^0_G$) given $\zb_D$ is:
\begin{align}
 \wtil_0&=\e(W_0|\Zb_{\Db}=\zb_{\Db})
    =(\ab_{0})^{T}(\Ab_{\Db}\Ab_{\Db}^{T}+\Psib_{\Db}^{2})^{-1} \zb_{\Db},
  \label{eq-bifactor-fs-W0} \\
 \wtil_g(\wtil_0)&=\e(W_g|\Zb_{\Db}=\zb_{\Db},W_0
 =\wtil_0) 
    =(\bm{b}_{g}^{T},0)(\sigb_{g})^{-1}(\zb_g^{T},\wtil_{0})^{T},
  \label{eq-bifactor-fs-Wg}
\end{align}
where $\ab_0$ is the first column of the loading matrix,
$\zb_g=(z_{1g},\ldots,z_{d_gg})^T$,
$\bd_{0g}$ and $\bd_g$ are the $d_g\times1$ global and local loading vector for group $g$; let $\bm{B}_{g}=[\bd_{0g},\bd_{g}]$ (matrix of size $d_{g}\times 2$). Let $\sigb_{g}$ be the correlation matrix of $(\Zb_{g}^{T},W_0)$. Then $\sigb_{g}=
\begin{bmatrix} 
\bm{B}_{g}\bm{B}_{g}^{T}+\Psib_{g}^2 & \bd_{0g}\\
\bd_{0g}^{T}&1\\
\end{bmatrix}
$
for $g\in\{1,\ldots,G\}$.
The proof that $\wtil_g(\wtil_0)=\e(W_g|\Zb_{\Db}=\zb_{\Db})$
is  given in the Appendix \ref{sec:equivalence}.

\noindent
\textbf{Oblique factor}:
Let $(w^0_1,\ldots,w^0_D,\zb_D)$ be a realization of 
$(W_1,\ldots,W_G,\Zb_D)$.
The proxies (or estimates of $w_1^0,\ldots,w^0_G$) given $\zb_D$ are:
\begin{equation}
    \wtil_{g}=\e(W_g|\Zb_g=\zb_g) 
 =\ab_g^{T}(\ab_g\ab_g^{T}+\Psib_g^2)^{-1}\zb_{g} ,
  \label{eq-oblique-fs}
\end{equation}
where $\zb_g=(z_{1g},\ldots,z_{d_gg})^T$, $\Psib_g$ is the $g$-th block diagonal of $\Psib$
and $\ab_g$ is the $d_g\times1$ loading vector for group $g$,
for $g\in\{1,\ldots,G\}$.
This version, rather than $\e(W_g|\Zb_h=\zb_h, h\in\{1,\ldots,G\})$,
has a version for the oblique factor copula that is numerically easier to handle.

\subsection{Proxies in factor copula models}
\label{sec-proxies-copula}

This subsection has the corresponding conditional expectations
of latent variables given $U(0,1)$ distributed observed variables
for the 1-factor, bi-factor, and oblique
factor copula models.

\noindent
\textbf{1-factor copula model}
with density $c_{\Ub_{D}}$ defined in 
\eqref{eq-1factor-copula-pdf}. Then
$c_{V\Ub_{\Db}}(v,\ub_{\Db})=\prod_{j=1}^{D}c_{jV}(u_j,v)$.
Let $(v^0,u_1,\ldots,u_{D})$ be one realization of $(V,U_1,\ldots,U_D)$.
Then the proxy, {as an estimated of $v^0$}, is:
\begin{equation}
  \label{eq-1factor-proxy}
 \vtil_{D}= \vtil_{D}(\ub_D) = \e(V|\Ub_{D}=\ub_{D})=
  \frac{\int_{0}^{1}vc_{V\Ub_{D}}(v,\ub_D)\,\dd v} {c_{\Ub_{D}}(\ub_{D})}.
\end{equation}

\noindent\textbf{Bi-factor copula model} with density as defined in
\eqref{eq-bifactor-copula-pdf}.
{Let $(v_0^0,\{v_g^0\},\{u_{jg}\})$ be one realization of
$(V_0,\{V_g\},\{U_{jg}\})$.}

Table \ref{tab:notations_bifact} has densities involving the
global latent variable $V_0$ and group latent variables $V_1,\ldots,V_G$.
\begin{table}[H]
\centering
\begin{tabular}{c|c}
\toprule
vector & joint density \\
\hline
$(\Ub_{\Db},V_{0},\Vb)$ & $\displaystyle
c_{\Ub_{\Db},V_{0},\Vb}(\ub_{\Db},v_{0},\vb)
=\prod_{g=1}^{G}\prod_{j=1}^{d_g}\Bigl\{c_{\Ugj V_{0}}(\ugj,v_{0})\cdot
c_{\Ugj \Vg;V_{0}}(C_{\Ugj|\Vo}(\ugj|v_0),v_g)\Bigr\}$ \\
\hline
$(\Ub_{\Db},\Vb_{0})$  & 
$\displaystyle  c_{\Ub_{\Db},\Vo}(\ub,\vo)= \prod_{g=1}^{G}\Bigl\{
\prod_{j=1}^{d_g}c_{\Ugj \Vo}(\ugj,\vo)\cdot f_{g}(\ub_{g};\vo) \Bigr\}$ \\
 &  $\displaystyle  f_{g}(\ub_{g};\vo)=\int_{0}^{1}\prod_{j=1}^{d_{g}}
c_{\Ugj\Vg;\Vo}\left(C_{\Ugj|\Vo}(\ugj|\vo),\vg \right)
 \dd \vg, \quad g\in\{1,\ldots,G\} $ \\
\hline
$(\Ub_{g},\Vb_{0},\Vb_{g})$ & $\displaystyle    
c_{\Ub_{g},\Vo,\Vg}(\ub_{g},\vo,\vg)=\prod_{j=1}^{d_g}
\Bigl\{c_{\Ugj\Vo}(\ugj,\vo)\cdot c_{\Ugj\Vg;\Vo}
\left(C_{\Ugj|\Vo}(\ugj|\vo),\vg\right) \Bigr\}$ \\

\bottomrule
\end{tabular}
\caption{Densities in the bi-factor copula model \eqref{eq-bifactor-copula-pdf}.}
\label{tab:notations_bifact}
\end{table}

The convenient form of conditional expectations is as follows:
\begin{equation}
    \vtil_{0\Db}(\ub_{\Db}) =\e[\Vo|\Ugj=\ugj, j\in\{1,\ldots,d_g\}, g\in\{1,\ldots,G\}]
  =\frac{\intt\vo c_{\Ub_{\Db},\Vo}(\ub_{\Db},\vo) \dd \vo}{\intt 
  c_{\Ub_{\Db},\Vo}(\ub_{\Db},\vo)\dd \vo}.
 \label{eq-bifactor-gproxy}
\end{equation}
For $g\in\{1,\ldots,G\}$,
\begin{eqnarray}
  \vtil_{g\Db}(\ub_g,{\vtil_{0\Db}})
   & = & \e[\Vg|\Vo=\vtil_{0\Db}, \Ugj=\ugj, j\in\{1,\ldots,d_{g}\}] \nonumber  \\
   & = & \frac{\int_{0}^{1}\vg c_{\Ub_{g},\Vo,\Vg}(\ub_{g},
   \vtil_{0\Db},\vg)\,\dd \vg}
 {\intt c_{\Ub_{g},\Vo,\Vg}(\ub_{g},\vtil_{0\Db},\vg)\,\dd \vg}.
\label{eq-bifactor-lproxy}
\end{eqnarray}
The proxies are $\vtil_{0\Db}(\ub_{\Db})$
and $\vtil_{g\Db}(\ub_g,\vtil_{0\Db})$ for $g\in\{1,\ldots,G\}$.

\noindent\textbf{Oblique factor copula model}
with density in \eqref{eq-oblique-copula-pdf}, for $g\in\{1,\ldots,G\}$. Then $c_{\Ub_{g},\Vg}(v_{g},\ub_{g})=\prod_{j=1}^{d_g}c_{\Ugj,V_{g}}(\ugj,v_g)$.
{Let $(\{v_g^0\},\{u_{jg}\})$ be one realization of
$(\{V_g\},\{U_{jg}\})$.}
\begin{equation}
 \vtil_{g\Db}(\ub_{g})=\e[V_{g}|\Ugj=u_{jg}, j\in\{1,\ldots,d_{g}\}] 
 =\frac{\intt v_g 
  c_{\Ub_{g},\Vg}(v_{g},\ub_{g})\,\dd v_{g}}{c_{\Ub_{g}}(\ub_{g})}
 =\frac{\intt v_g c_{\Ub_{g},\Vg}(v_{g},\ub_{g})\,\dd v_{g}}
  { \intt c_{\Ub_{g},\Vg}(v_{g},\ub_{g})\,\dd v_{g}}
  \label{eq-oblique-proxy}.
\end{equation} 
{The proxy for $v_g^0$ is $\vtil_{g\Db}(\ub_{g})$ for $g\in\{1,\ldots,G\}$.}

\section{Consistency of proxies and rate of convergence: model known}
\label{sec-consistency1}

In this section, we obtain conditions so that the
proxies defined in Section \ref{sec-proxies}
are consistent as $D\to\infty$ for 1-factor and
$d_g\to\infty$ for all $g$ for bi-factor or oblique factor models.
More direct calculations are possible for Gaussian models
and these provide insights into behavior for factor copulas.

\subsection{Conditional variance for Gaussian factor models}
\label{sec-condvar}

We start with the conditional variance of the latent variables given the
observed variables. If the conditional variance does not go
to zero {as $D\to\infty$}, then the latent variable cannot be consistently
estimated; this can happen if the overall dependence with the latent
variable is weak, even as more variables are added. 
If the conditional variance is 0 for a finite $D$ or $\Db$,
then the latent variable can be determined exactly (this can happen for
the 1-factor model if $\psi_j=0$ for some $j$).
The practical case is when the dependence is moderate to strong,
so that intuitively we have a better idea of the value of the latent
variable as $D$ or $\Db$ increases.

We summarize the expressions of the conditional variances in Gaussian
factor models in Table \ref{tab:cond_var}.
The 1-factor model
and bi-factor model are special cases of the $p$-factor model.
For the bi-factor model, we decompose the conditional variance of
group latent variables into two parts; one part only depends on the
within-group dependence, and another part comes from the conditional
variance of the global latent variable. 

Details of the derivations for
the decomposition formula in \eqref{decomp_0} of the conditional variance can be found in the
Appendix \ref{sec: decomp}. The expression of the conditional variance in the oblique factor
model is slightly different from the $p$-factor models, but the deviation procedures are similar so the details are omitted.

\begin{table}[H]
\centering
\begin{tabular}{l|l} 
\toprule
 model & conditional variance\\
 \hline
$p$-factor & $\cov(\Wb|\Zb_{D}=\zb_{D})=
 \bigi_{p}-\Ab_{D}^{T}(\Ab_{D}\Ab_{D}^{T}+\Psib^{2}_{D})^{-1}\Ab_{D}$ \\
 & $=\bigi_p-(\bigi_p+\Ab_D^T\Psib_D^{-2}\Ab_D)^{-1}
  \Ab_D^T \Psib_D^{-2}\Ab_D$ \\
\hline
bi-factor & $\Var(W_0|\Zb_{\Db}=\zb_{\Db})
   = 1-\ab_0^{T}(\Ab_{\Db}\Ab_{\Db}^{T}+\Psib^{2}_{\Db})^{-1}\ab_0$ \\
 & \parbox{3cm}{ \begin{equation}
\label{decomp_0}
\Var(W_g|\Zb_{\Db}=\zb_{\Db}) =
(1+q_{g})^{-1}+\qtil_{g}^{2}(1+q_g)^{-2}\Var(W_0|\Zb_{\Db}=\zb_{\Db})
\end{equation}}
\\
\hline
oblique factor &  
\parbox{3cm}{
\begin{equation}
\label{eq:obl_condvar}
\cov(\Wb|\Zb=\zb_{\Db}) =
\bm{\Sigma}_{\Wb}-\bm{\Sigma}_{\Wb}\Ab_{\Db}^{T}(\Ab_{\Db}\bm{\Sigma}_{\Wb}
\Ab_{\Db}^{T}+\Psib_{\Db}^{2})^{-1}\Ab_{\Db}\bm{\Sigma}_{\Wb}  
\end{equation}} \\  
\bottomrule
\end{tabular}
\caption{Conditional variance of latent factors in general and structured
Gaussian factor models; the matrix equality for $p$-factor follows
from \eqref{eq-matrix-id}; 
$q_g=\bd_g^{T}\Psib_g^{-2}\bd_g$,
$\qtil_g=\bd_g^{T}\Psib_g^{-2}\bd_{0g}$, {$\bd_g$, $\bd_{0g}$ are the $d_{g}\times 1$ local and global loading vector for group $g$.}}
\label{tab:cond_var}
\end{table}

\textbf{Limit of covariance for $p$-factor}\quad 
In the (second) expression of conditional variance for the general $p$-factor
models, define $\Qb_{D}=\Ab_{D}^{T}\Psib_{D}^{-2}\Ab_{D}$ as a $p\times p$
matrix.
Suppose $\bar{\Qb}_{D}=: D^{-1}\Qb_{D}\to{\Qb}$  as $D\to\infty$.
Note that if dependence in the loading matrices is weak, then $\Qb$ can be
the zero matrix. In the boundary case with $\Ab_D=\ob$,
then $\cov(\Wb|\Zb_{D}=\zb_{D})=\bigi_p$, that is, $\Zb_D$ provides
no information about $\Wb$.
If $\Qb$ is a positive definite matrix, then
 $$\bigi_p-(\bigi_p+\Ab_D^T\Psib_D^{-2}\Ab_D)^{-1}\Ab_D^T\Psib_D^{-2}\Ab_D
  \approx \bigi_p-(\bigi_p+D\Qb)^{-1}D\Qb
  = \bigi_p- (D^{-1}\Qb^{-1}+\bigi_p)^{-1}
  = D^{-1}\Qb^{-1} +o(D^{-1}).$$

In the $p$-factor model, if the loading matrices in $\{\Ab_D\}$
have full column rank, and
the strength of dependence between
observed variables and the latent factors is strong enough
such that the limiting matrix $\Qb$ is invertible, then
the limit of conditional covariance for $p$-factor is $O(D^{-1})$. 
{The next theorem indicates what happens if the condition of
full column rank does not hold.}

\begin{theorem}
\label{1_factor_gaussian_2}
Consider the $p$-factor model with $p\ge2$ with matrix  representation
\eqref{eq-factor_mat}. \\
(a)
If $\Ab_D$ does not have full column rank, then the latent variables
in $\Wb$ are not identifiable. \\
(b) For the Gaussian bi-factor model as a special case of $p$-factor with $p=G+1$,
if $\Ab_D$ does not have full column rank, then the Gaussian bi-factor model
can be rewritten as an oblique factor model with fewer parameters.
\end{theorem}

\begin{proof}
(a) Let $\ab_1,\ldots,\ab_p$ be the columns of $\Ab_D$.
The columns of $\Ab_{D}$ are linearly dependent.
Without loss of generality, assume $\ab_1=t_2\ab_2+\cdots+t_p\ab_p$
where $(t_2,\ldots,t_p)$ is a non-zero vector.
Then, in \eqref{eq-factor_mat},
  $$\Zb_D-\Psib_D\epsb_D
  =\sum_{j=2}^p t_j\ab_j \cdot W_1 + \sum_{j=2}^{p}\ab_j W_j
  = \sum_{j=2}^p \ab_j (t_jW_1+W_j).$$
Hence only some linear combinations of the latent variables can
be identified. \\
(b) For bi-factor, let  $\ab_0,\ab_1,\ldots,\ab_G$ be the columns of $\Ab_D$,
and let the latent variables be $W_0,W_1,\ldots,W_G$.
The above implies that
  $$\Zb_D = \sum_{g=1}^{G} {\ab_g} W_g^* +\Psib_D\epsb_D= \Ab^*_D\Wb^* +\Psib_D\epsb_D,$$
where $W_g^*= t_gW_0+W_g$ for a non-zero vector $(t_1,\ldots,t_G)$
and $\Ab^*_D=(\ab_1,\ldots,\ab_G)$ is $D\times G$.
The identifiable latent variables $W_g^*$ are dependent.
\end{proof}

\begin{remark}

For the bi-factor model, if the global loading vector is roughly
equal to a linear combination of the group loading vectors, then the
latent factors are close to non-identifiable, and the
oblique factor model may be a good approximation. A useful diagnostic
tool is the condition number of $\Qb_D = \Ab_D^T \Psib_D^{-2}\Ab_D$
because $\Qb_D$ is not of full rank if $\Ab_D$ is not of full column rank.
If the condition number is small enough, then the bi-factor model is
appropriate to use; otherwise, oblique factor model can be a good fit.
\end{remark}

For the oblique factor model, with $\Bb_{\Db}=\Ab_{\Db}\sigb_{\Wb}^{1/2}$
and using \eqref{eq-matrix-id}, 
the right-hand-side of
\eqref{eq:obl_condvar} becomes:

{
\begin{equation}
 \sigb_{\Wb}-\sigb_{\Wb}^{1/2}\Bb_{\Db}^{T}(\Bb_{\Db}
\Bb_{\Db}^{T}+\Psib_{\Db}^{2})^{-1}\Bb_{\Db}\sigb_{\Wb}^{1/2}
  =  \sigb_{\Wb}-\sigb_{\Wb}^{1/2}(\bigi_p + \Bb_{\Db}^{T} \Psib_{\Db}^{-2}
  \Bb_{\Db})^{-1} \Bb_{\Db}^{T} \Psib_{\Db}^{-2} \Bb_{\Db} \sigb_{\Wb}^{1/2},
 \label{eq:obl_condvar2}
\end{equation}
}
and
 $$\Bb_{\Db}^{T} \Psib_{\Db}^{-2} \Bb_{\Db}=
  \sigb_{\Wb}^{1/2}\Ab_{\Db}^{T}\Psib_{\Db}^{-2} \Ab_{\Db}
  \sigb_{\Wb}^{1/2} = \sigb_{\Wb}^{1/2} \Qb_{\Db}\sigb_{\Wb}^{1/2}.$$
After some algebraic calculations, this simplifies to {$\sigb_{\Wb}-(\bigi_p+D^{-1}\Qb^{-1}\sigb_{\Wb}^{-1})^{-1}\sigb_{\Wb}$.}
The limit of conditional covariance is $O(D^{-1})$ if
$D^{-1}\Qb_{\Db}\to \Qb$ with $\Qb$ being diagonal and non-singular.

\medskip

\noindent
\textbf{Rate of convergence as a function of strength of
dependence with latent variables}.

\begin{itemize}
\item For 1-factor model, $\Qb_{D}$ is $1\times 1$ and can be denoted by the
scalar $q_{D}$. Assume that
$\bar{q}_{D}=D^{-1}q_{D}:=D^{-1}\sum_{j=1}^{D}\alpha_{j}^2/(1-\alpha_{j}^2)\to
q>0$, so that the dominating term of the conditional variance is 
$D^{-1}{q}^{-1}$. 
If one or several entries of loadings increase, ${q}^{-1}$ decreases; that is, the convergence will be faster.

\item The bi-factor case with a full-rank $\Ab_D$ loading
matrix is more complicated. We consider several cases
to understand how the dependence affects the conditional
variance, and all the claims are inferred from conditional
variance decomposition formula \eqref{decomp_0} shown in Table \ref{tab:cond_var}
and numerical experiments. 

\textbf{case a)} If the global loadings are constant, then the
group with stronger dependence has a smaller conditional variance, thus a faster convergence rate. From
\eqref{decomp_0}, the conditional variance of local latent factors
will be most affected by the first term because in the second term,
the conditional variance of $W_0$ is fixed for all groups, and the ratio
$\qtil_g^{2}/(1+q_g)^{2}$ will not differ much.

\textbf{case b)} Suppose $\Ab_D$ is a well-conditioned matrix under small
perturbations. 
If the global loadings are fixed, and the local
loadings for one group, for example, $g'$ increase a little. From
the numerical experiments, the conditional variance for $W_g$, $g=0$ 
or $g\neq g'$ could increase or decrease a little, but
$\Var(W_{g'}|\Zb_{\Db})$ will decrease. The argument can
also be inferred from \eqref{decomp_0}: 
as $\bd_{g'}$ increases, the first term decreases, and it
dominates the change of conditional variance.

\textbf{case c)} Suppose $\Ab_D$ is a well-conditioned matrix under small
perturbations.
If the global loadings increase a little but the local loadings remain the
same, $\Var(W_0|\Zb_{\Db})$ will decrease according to the
numerical results.

\item In the oblique factor model, assuming that the groups are
similar in size and increase to infinity, with no group dominating,
the conditional variance of the latent variables is dominated
by the term $D^{-1}{\Qb}^{-1}$, where $\Qb$ is 
diagonal and $(g,g)$ entry being the limit of
$D^{-1}\sum_{j=1}^{d_g}\alpha_{jg}^{2}/\psi_{jg}^{2}$ as $\Db\to\infty$. 
 {Therefore},
the conditional covariance of latent variables is closely
related to the strength of within-group dependence. The proxy variable
in the group with relatively strong dependence will have smaller
conditional variance.

\end{itemize}
Suppose the conditional variance of a latent variable given the observed
variables does not go to 0 as $D\to\infty$, then one cannot expect the 
corresponding proxy estimates in Section 2 to be consistent. For factor copula models, the
conditional variances of the latent variables do not have closed forms. The results in this section for Gaussian factor models provide
insights into conditions for consistency of proxy estimates in factor copula models as well as the connection between the rate of convergence and strength of dependence with the latent variables. Under some regularity conditions on the bivariate linking copulas, it can be shown that the limit of the conditional variance in factor copula models is also $O(D^{-1})$.

\subsection{Consistency in Gaussian factor models}
\label{sec-consist-gauss}

In this subsection, mild conditions are obtained for consistency
of proxy estimates via conditional expectations in Section \ref{sec-proxies}.
The cases that are covered in the theorems have moderate to strong
dependence, without loading parameters going to $\pm1$ as $D\to\infty$.
In the latter case, with even stronger dependence with latent variables,
there is consistency, but the proofs would be different because
identity \eqref{eq-matrix-id} would not hold in the limit.
The conditions in the theorems match practical uses of factor models ---
one might have idea of latent factors that affect dependence within
groups of variables; there is at least moderate dependence
among observed variables and dependence is not so strong that one
variable could be considered as a proxy for the latent variable. 

\begin{theorem}(Asymptotic properties of factor scores in 1-factor Gaussian model)
\label{1_factor_gaussian}
For \eqref{eq-1factor-gauss},
suppose there is a realized infinite sequence of observed variables $z_1,z_2,\ldots$ and a realized infinite sequence of disturbance terms $e_1,e_2,
\ldots$ with realized latent variable $w^{0}$ (independent of
dimension $D$) from the 1-factor model. 
For the truncated sequence to
the first $D$ variables, let $\zb_{D}=(z_1,\ldots,z_{D})^{T}$,
$\bm{e}_{D}=(e_1,\ldots,e_{D})^{T}$, and 
let the loading matrix or vector be $(\alpha_1,\ldots,\alpha_D)^T$. 
Assume
\begin{equation*}
  -1<\liminf_{j\to\infty} \alpha_{j}<\limsup_{j\to\infty} \alpha_{j}<1, 
  \quad   \lim_{D\to\infty}D^{-1}\sum_{j=1}^{D}|\alpha_{j}|\to \text{const}, 
  \quad \text{const}\neq 0.
\end{equation*}
Then for the factor scores defined in equation \eqref{eq-1factor-fs}, 
$\wtil_{D}-w^{0}=O_{p}(D^{-1/2})$ as $D\to\infty$.
\end{theorem}

\begin{remark}
The proof is given in the Appendix \ref{sec: theorem_1_fact}. 
The assumption on $\alpha_j$ uniformly bounded away
from $\pm 1$ ensures that the proxies are well-defined in two equivalent
forms in \eqref{eq-1factor-fs}. The second assumption about the
averaged absolute loadings ensures that the dependence is
strong enough, because from Section \ref{sec-condvar},
consistency does not hold in the case of sufficiently weak dependence.
\end{remark}

\begin{corollary}(Asymptotic properties of factor scores in oblique Gaussian model).
\label{oblique_fact_consistent}
Suppose there are realized infinite sequences of observed
variables $\zb_{1}^{T}={(z_{1,1},z_{2,1},\ldots)}$,
$\ldots,\zb_{G}^{T}={(z_{1,G},z_{2,G},\ldots)}$
and realized sequences
of disturbance terms
$\bm{e}_{1}^{T}={(e_{1,1},e_{2,1},\ldots)}$,
$\ldots,\bm{e}_{G}^{T}={(e_{1,G},e_{2,G},\ldots)}$,
with latent variables $\wb^{0}=(w_{1}^{0},\ldots,w_{G}^{0})$
from the oblique factor model with fixed $G$ groups
defined in \eqref{eq-oblique-gauss}. 
Truncate the
sequences to the first $d_g$ variables $\zb_{g,d_g}$ for
$g\in \{1,\ldots,G\}$ with no $d_g$ dominating others. Let
$\zb_{\Db}=(\zb_{1,d_1}^{T},\zb_{2,d_2}^{T},\ldots,\zb_{G,d_G}^{T})^{T}$,
$\bm{e}_{\Db}=(\bm{e}_{1,d_1}^{T},\bm{e}_{2,d_2}^{T},\ldots,\bm{e}_{G,d_G}^{T})$.
Assume $d_g^{-1}\|\ab_{g}\|_{1}\not\to 0$ as $d_g\to\infty$
for $g\in\{1,2,\ldots,G\}$.
Let $\wtilb_{\Db}=(\wtil_1,\ldots,\wtil_{G})$ be the factor scores defined in \eqref{eq-oblique-fs}. Then 
\begin{equation*}
  \wtil_{g}-w_{g}^{0}=O_{p}(D^{-1/2}), \quad  g\in\{1,\ldots,G\}.
\end{equation*}
\end{corollary}

The consistency of proxies in the oblique factor model is a
straightforward corollary of Theorem \ref{1_factor_gaussian}. Similar
to the 1-factor model, the assumptions on the strength of dependence in
the model suggests the with-in group dependence is not weak.

\begin{theorem}{(Asymptotic properties of factor scores in bi-factor Gaussian model)}.
\label{bi_fact_consistent}
Suppose there are realized infinite sequences of observed
variables $\zb_{1}^{T}={(z_{1,1},z_{2,1},\ldots)}$,
$\ldots,\zb_{G}^{T}={(z_{1,G},z_{2,G},\ldots)}$ and realized sequences
of disturbance terms
$\bm{e}_{1}^{T}={(e_{1,1},e_{2,1},\ldots)}$,
$\ldots,\bm{e}_{G}^{T}={(e_{1,G},e_{2,G},\ldots)}$,
with latent variables $\wb^{0}=(w_{0}^{0},w_{1}^{0},\ldots,w_{G}^{0})$
from the bi-factor model with fixed $G$ groups defined
in \eqref{eq-bifactor-gauss}. Truncate the sequences to
the first $d_g$ variables $\zb_{g,d_g}$ for $g\in \{1,\ldots,G\}$, with no
$d_g$ dominating others. Let
$\zb_{\Db}=(\zb_{1,d_1}^{T},\zb_{2,d_2}^{T},\ldots,\zb_{G,d_G}^{T})^{T}$,
$\bm{e}_{\Db}=(\bm{e}_{1,d_1}^{T},\bm{e}_{2,d_2}^{T},\ldots,\bm{e}_{G,d_G}^{T})$
and assume that the loading matrices
$\Ab_{\Db}=[\ab_{0},\diag(\ab_{1},\ldots,\ab_{G})]$ 
are of full rank, with bounded condition number over $d_g\to\infty$ for all $g$.
$d_g^{-1}\|\ab_{j}\|_{1}\not\to 0$ for $g\in \{0,1,\ldots,G\}$.
Let $\wtilb_{\Db}=(\wtil_{0},\ldots,\wtil_{G})$ 
be the factor scores defined in \eqref{eq-bifactor-fs-W0}
and \eqref{eq-bifactor-fs-Wg}. Then 
\begin{equation*}
 \wtil_{g}-w_{g}^{0}=O_{p}(D^{-1/2}), \quad  g \in \{0,1,2,\ldots,G\}.
\end{equation*}
\end{theorem}

\subsection{Consistency in factor copula models}
\label{sec-consist-copula}

In this section, we state results with mild conditions
for the consistency of the proxy variables in the
factor copula models with known parameters. 
The conditions and interpretation parallel those in the preceding Section
\ref{sec-consist-gauss}.

We next state some assumptions that are assumed throughout this section. 

\begin{enumerate}[label=\textbf{Assumption} \arabic{enumi}.,ref=Assumption \arabic{enumi}, wide=0pt]
\item \label{bivariate}
\begin{enumerate}
\item The bivariate linking copulas have monotonic dependence, that is,
the observed variables are
monotonically related to the latent variables.
\item For any fixed dimension $\Db$, the log-likelihood function {with latent variables considered as parameters to be estimated},
satisfy some standard regularity conditions,
such as in \cite{bradley1962asymptotic}.
{For example}, continuity of derivatives up to third-order of the log-densities of the bivariate linking copulas with respect to $v$'s. 
\end{enumerate}
\end{enumerate}
For the conditional expectations for factor copula models, the $v$'s are treated as parameters and the $u$'s are realization of independent random variables when the latent variable are fixed.
The proofs make use of the Laplace approximation method.

Some results in \cite{krupskii2022approximate} assume the
observed variables are stochastically increasing in the latent variables,
and this implies observed variables are  monotonically related.
The above consists of a mild condition, because one would not think of
using factor models with variables are not monotonically related.

\begin{theorem}
\label{consistent_proxy_1fact}
{(Consistency of proxy in 1-factor copula model with known linking
copulas)}
Suppose there is a realized infinite sequence $u_1,u_2,\ldots,$ with
latent variable {$v^{0}$} (independent of dimension $D$)
from the 1-factor model in \eqref{eq-1factor-copula-pdf}.
For the truncation to the first $D$ variables, 
let $\ub_{D}=(u_1,\ldots,u_{D})^{T}$.
Define the averaged negative log-likelihood in parameter $v$ as
\begin{equation*}
     \bar{L}_{D}(v)=-D^{-1}\sum_{j=1}^{D}\log c_{jV}(u_{j},v).
\end{equation*}
Assume $\lim_{D\to\infty} \bar{L}_{D}$ asymptotically has a global
minimum and is strictly locally convex at the minimum. Also, assume the
likelihood function satisfies the usual regularity conditions.
Consider the proxy defined in \eqref{eq-1factor-proxy}. 
As $D\to\infty$, then
$\vtil_{D}-{v^{0}}=O_{p}(D^{-1/2})$ as $D\to\infty$.
\end{theorem}

\begin{remark}
The theorem shows that under certain regularity conditions, the latent
variables can be approximately recovered from the observed variables,
assuming that the number of variables monotonically linked to the
latent variable is large enough and that the dependence is
strong enough. If the overall dependence of the $c_{jV}$ is weak, with 
many copulas approaching independence ($c_{jV}(u,v)\approx u$ for many $j$), 
then it is possible that $\bar{L}_D(v)$ is a constant function in the limit. Our assumption on the limiting function of the averaged
log-likelihood is mild; if all the bivariate copulas are strictly
stochastically increasing, it's not hard to show the limiting function
is locally convex around the true realized value {$v^{0}$} and {$v^{0}$}
is a global minimum of the function.
If the $U_j$ are not monotonically related to the latent variable,
then it is possible for 
$\bar{L}_{D}(v)$ to have more than one local minimum.
An example consists of:
$C_{jV}$ is the copula of $(U_j,V)$ such that
(i) $(U_j,2V)$ follows the Gaussian copula with parameter $\rho_j>0$ if
$1/2\le V<1$ and
(ii) $(U_j,2V)$ follows the Gaussian copula with parameter $-\rho_j<0$ if
$0<V<1/2$,
and the $\{\rho_j\}$ is uniformly distributed in a bounded interval such as $[0.2,0.8]$.
\end{remark}

\noindent\textbf{Oblique factor copula model with known linking copulas}:
In the oblique factor copula model, the variables in one group are linked
to the same latent variable and these variables satisfy a 1-factor copula model. Hence, the
assumptions and conclusion for the 1-factor copula model 
extend to the oblique factor model.

\begin{corollary}(Consistency of proxies in oblique factor copula model with known linking copulas).
\label{consistency_proxy_obl}
Suppose there are realized infinite sequences of observed variable
values $\ub_1^{T},\ub_2^{T},\ldots,\ub_G^{T}$ with latent
variable values $\vb^{0}=(v_{1}^{0},v_2^{0},\ldots,v_G^{0})^{T}$
from the oblique factor model with $G$ groups defined in
equation \eqref{eq-oblique-copula-pdf}.
Truncate the sequences
to the first $d_g$ variables $u_{g,d_g}$ for $g\in\{1,\ldots,G\}$
with no $d_g$ dominating.
Let
$\ub_{\Db}=(\ub_{1,d_1}^{T},\ub_{2,d_2}^{T},\ldots,\ub_{G,d_G}^{T})^{T}$. 
With $D=\sum_{g=1}^G d_g$, define the averaged negative log-likelihood in
group $g$ as 
 $\bar{L}_{D}^{(g)}(v_g)=-D^{-1}\sum_{j=1}^{d_g}\log c_{U_{jg},V_g}(u_{jg},v_g)$. 
Assume $\lim_{d_g\to\infty\forall g}\bar{L}_{D}^{(g)}$
asymptotically has a global minimum and is strictly locally {convex} at the
minimum. Also assume these likelihood functions satisfy the usual regularity
conditions. Then, consider proxy variable defined in \eqref{eq-oblique-proxy},
as $d_g\to\infty$, $\vtil_{g\Db}-v_g^{0}=O_{p}(D^{-1/2})$ for
$g\in\{1,2,\ldots,G\}$.
\end{corollary}

\begin{theorem}
{(Consistency of proxies in bi-factor copula model with known linking copulas)}.
\label{consistency_proxy_bifct} Suppose there are realized infinite sequences of observed variable
values $\ub_1^{T},\ub_2^{T},\ldots,\ub_G^{T}$ with latent
variables values $\vb^{0}=(v_0^0,v_{1}^{0},v_2^{0},\ldots,v_G^{0})^{T}$
from the bi-factor model with $G$ groups defined in
\eqref{eq-bifactor-copula-pdf}. Truncate the
sequences to the first $d_g$ variables $\ub_{g,d_g}$ for
$g\in\{1,\ldots,G\}$. 
Let $\ub_{\Db}=(\ub_{1,d_1}^{T},\ldots,\ub_{G,d_G}^{T})^{T}$. Let
$L_{0}(\vo;\ub_{\Db})=\log c_{\Ub_{\Db},\Vo}(\ub_{\Db},\vo)$
be a log-likelihood function in $\vo$ with observed variables
$\ub_{\Db}$. For $g\in\{1,\ldots,G\}$, let 
$L_{g}(\vg;\vo,\ub_{g,d_g})=\log c_{\Ub_g,\Vo,\Vg}(\ub_{g,d_g},\vo,\vg)$ 
be a log-likelihood function
in $\vg$ with observed $\ub_{g,d_g}$ and given $\vo$.
Define the averaged negative log-likelihood
for marginalized density of $(\Ub_{\Db},\Vo)$ as
$\bar{L}_{0\Db}(\vo;\ub_{\Db})=-D^{-1}L_{0}(\vo;\ub_{\Db})$
and the averaged negative log-likelihood for
marginalized density of $(\Ub_{g,d_g},\Vo,\Vg)$ as
$\bar{L}_{g\Db}(\vg;\vo,\ub_{g,d_g})=-D^{-1}L_{g}(\vg;\vo,\ub_{g,d_g})$. Assume
all the likelihood functions satisfy the usual regularity conditions.
Assume $\lim_{d_{g}\to\infty}
\bar{L}_{g\Db}$ asymptotically has a global minimum and is strictly
locally {convex} at the minimum. The same assumptions are {applied} to
$\lim_{\Db\to\infty}\bar{L}_{0\Db}$. Then, $\vtil_{0\Db}$
defined in \eqref{eq-bifactor-gproxy} and  $\vtil_{g\Db}(\vtil_{0\Db})$ for
$g=1,\ldots,G$ defined in \eqref{eq-bifactor-lproxy}  are consistent for
$\vo^{0},v_{1}^{0},\ldots,v_{G}^{0}$ respectively as $d_g\to\infty$
for all $g$.
\end{theorem}

\section{Consistency of proxies with estimated parameters}
\label{sec-consistency2}

With a parametric model for the loading matrix in the Gaussian
factor models, and parametric bivariate linking copulas in the
factor copula models, the parameters can be estimated and then the
proxies in Section \ref{sec-proxies} can be applied with the usual
plug-in method. For factor copulas, numerical integration would
be needed to evaluate the proxies for a random sample of size $N$.

In this section, we prove the consistency of proxies with estimated parameters 
in two steps.
For step 1, we prove the equations for the proxy variables are 
locally Lipschitz in the parameters (even as
number of parameters increase as $D$ increases). 
For step 2, under the assumption that the observed variables can be considered
as a sample from a super-population, all parameters can be estimated
with $O(N^{-1/2})$ accuracy.

The theoretical results of this section and the preceding section support
the use of proxies as estimates of latent variables for $D$ large enough,
and $N$ large enough in factor models.
The next Section \ref{sec-sequential} outlines a sequential method
for determining proxies for factor copulas.

\subsection{Gaussian factor models}

\begin{lemma}
\label{lip_factormodel}
In the $p$-factor model
{\eqref{eq-factor_mat}},
let the loading matrix be
${\Ab}^{T}=(\ab_1,
\ldots,\ab_D,\ldots)$, $\widehat{\Ab}^{T}=(\hataab_1,
\ldots,\hataab_D,\ldots)$ (infinite sequence) be the perturbation of
${\Ab}^{T}$, where $\hat{\bm{a}}_{j}$ and $\bm{a}_{j}$
denotes the $j$th column of matrix $\hatAb^{T}$ and $\Ab^{T}$
respectively,
$j\in\{1,2,\ldots\}$. Suppose there is one realization
$\zb=(z_1,\ldots,z_{D},\ldots)$ from the factor model with
loading matrix $\Ab$. Let $\Ab_{D}$ be $\Ab$ truncated to
$D$ rows, and similarly define $\widehat{\Ab}_{D}$.
Suppose the entries of ${\Psib}_{D}^{2}, 
\widehat{\Psib}_{D}^{2}$ are bounded away from 0.
Consider the factor scores vector in \eqref{eq-1factor-fs} as a function of the loading matrix: 
$\wtilb_{D}=\wtilb_{D}(\Ab_{D})$.
Let $\Qb_{D}=\Ab_{D}^{T}\Psib_{D}^{-2}\Ab_{D}$,
$\widehat{\Qb}_{D}=\widehat{\Ab}_{D}^{T}\widehat{\Psib}_{D}^{-2}
\widehat{\Ab}_{D}$.
Suppose $D^{-1}{\Qb_{D}}\to \Qb$, $D^{-1}{\widehat{\Qb}_{D}}\to
\widehat{\Qb}$ where $\Qb,\widehat{\Qb}$ are both positive definite
and well-conditioned matrices.
Then 
$\|\wtilb_{D}(\widehat{\Ab}_{D})-\wtilb_{D}(\Ab_{D})\|\leq
K_{D}\cdot
\sqrt{D^{-1}\sum_{j=1}^{D}\|\widehat{\ab}_j-\ab_j\|^2}$, where the
constants $K_{D}$ are bounded as $D\to\infty$.
\end{lemma}

\subsection{Parametric copula factor models}

For parametric models we assume that there is parameter associated
with each bivariate linking copula, so that generically, the copulas of
form $c_{U_s,V}(u,v)$ in Section \ref{sec-factor} are now written
as $c_{U_s,V}(u,v;\btheta_s)$.

\begin{lemma}
\label{lipschitz}
Consider a 1-factor copula model \eqref{eq-1factor-copula-pdf} with parametric
linking copulas that have monotone dependence {and the bivariate linking copulas satisfy \ref{bivariate}}. Let the
parameter vector of linking copulas be the infinite sequence
$\bm{\theta}=(\btheta_1,\btheta_2,\ldots,\btheta_{D},\ldots)$. Let $\hat{\btheta}$ be a perturbation of $\btheta$.  
Assume the parameters of both $\btheta,\hat{\btheta}$
all lie in a bounded space such that
all the linking copulas are bounded away from comonotonicity
and countermonotonicity. Let $\ub=(u_1,\ldots,u_{D},\ldots)$
be one realization generated from model with $\bm{\theta}$. Let
$\btheta_D$ (or $\hat{\btheta}_D$), $\ub_D$ be truncated to the first $D$ linking copulas 
and variables.
Consider the proxy in \eqref{eq-1factor-proxy} as a function
of $\btheta_D$: $\vtil_{D}=\vtil_{D}(\btheta_{D})$. 
Suppose $\hat{\btheta}_{D}\in
\bar{\boldsymbol{B}}(\btheta_{D},\rho)$ (ball of sufficiently small radius 
$\rho>0$).
Assume for each $j$, the following partial derivatives of the log copula 
densities exist:
$\partial^{k}\log c_{jV}(u_j,v;\btheta_j) /\partial v^{k}$, $k\in\{1,2,3\}$,
and  $\partial^{k+1}
\log c_{jV}(u_j,v;\btheta_j)/\partial\btheta_j\partial v^{k}$ for $k\in\{0,1,2\}$. 
Also, assume the derivatives $\partial^{k+1} \log
c_{jV}(u_j,v;\btheta_j)/\partial \btheta_j\partial v^{k}$ are uniformly bounded
for $k\in\{0,1,2\}$ in $\bar{\bm{B}}(\btheta_{D},\rho)$.
Then there exist a
constant $B_{D}$ which is bounded as $D\to\infty$ such that
\begin{equation*}
    \|\vtil_{D}({\hat{\btheta}_{D}})-\vtil_{D}({\btheta}_{D})\|
  \leq B_{D}\|\hat{\btheta}_{D}-\btheta_{D} \|^{*}
\end{equation*}
where
$\|\hat{\btheta}_{D}-\btheta_{D}\|^{*}:=\sqrt{D^{-1}
\sum_{j=1}^{D}\|\hat{\btheta}_{j}-\btheta_j\|_2^2}$.
\end{lemma}

\begin{lemma}
\label{lipschitz2}
Consider a $G$-group bi-factor copula model \eqref{eq-bifactor-copula-pdf} 
with parametric linking copulas that have monotone dependence.
Let the global linking copula densities be
$c_{\Ugj,V_0}(\ugj,v_0;\btheta_{jg,0})$ in group $g$, $j\in\{1,2,\ldots\}$.
Let $\btheta_{g}^{(1)}=(\btheta_{1g,0},\btheta_{2g,0},\ldots)$. Let the local linking copula densities be $c_{\Ugj V_{g};V_0}(C_{\Ugj|V_0}(
\ugj|v_0),v_g;\btheta_{jg})$ in group $g$, $j\in\{1,2,\ldots\}$.
Let $\btheta_{g}^{(2)}=({\btheta}_{1g},{\btheta}_{2g},\ldots)$ be the vector
of parameters for $g\in\{1,\ldots,G\}$. Let
$\btheta^{(1)}=({\btheta}_{1}^{(1)},\ldots, {\btheta}_{G}^{(1)})$
and
$\btheta^{(2)}=({\btheta}_{1}^{(2)},\ldots, {\btheta}_{G}^{(2)})$
with perturbations $\hat{\btheta}^{(1)}$ and $\hat{\btheta}^{(2)}$.
Assume the parameters in these four vectors are all in a
bounded space such that all the linking copulas are bounded
away from comonotonicity and countermonotonicty.
Let $\Db=(d_1,\ldots,d_G)$ and
$\ub_{\Db}^{T}=(\ub_1^{T},\ub_2^{T},\ldots,\ub_{G}^{T})$ be a
truncation of infinite-dimensional realization from the model with
parameters 
$\btheta_{\Db}=(\btheta_{\Db}^{(1)},\btheta_{\Db}^{(2)})$, where
$\btheta_{\Db}^{(1)}$,$\btheta_{\Db}^{(2)}$ consist of $\btheta^{(1)}$,
$\btheta^{(2)}$ truncated to the first $d_g$ in group $g$. Let $\btheta_{g,d_{g}}=(\btheta_{g,d_g}^{(1)},\btheta_{g,d_g}^{(2)})$, where
$\btheta_{g,d_g}^{(1)}$,$\btheta_{g,d_g}^{(2)}$ consist of $\btheta_g^{(1)}$,
$\btheta_g^{(2)}$ truncated to the first $d_g$ in group $g$, for $g\in\{1,\ldots,G\}$. Similarly, define  $\hat{\btheta}_{\Db},\hat{\btheta}_{g,d_g}$. Let $\vtil_{0\Db}=\vtil_{0\Db}(\btheta_{\Db})$,
$\vtil_{g\Db}=\vtil_{g\Db}(\btheta_{\Db})$ be the proxies
defined in \eqref{eq-bifactor-gproxy} and \eqref{eq-bifactor-lproxy}.
Assume $\hat{\btheta}_{\Db}\in \bar{\bm{B}}(\btheta_{\Db},\rho)$
(ball of sufficiently small radius $\rho>0$).
Assume for $j\in\{1,\ldots,d_g\}$, $g\in\{1,\ldots,G\}$, that the
following partial derivatives of log copula densities exist:
$\partial^{k} \log c_{\Ub_{\Db},V_0}(\ub_{\Db},v_0,\btheta_{\Db})/\partial v_0^{k}$,
$\partial^k\log c_{\Ub_{g},V_g;V_0}(\ub_g,v_g;v_0,{\btheta_{g,d_g}})/
\partial v_g^{k}$, for $k\in\{1,2,3\}$ and $\partial^{k+1}
\log c_{\Ub_{\Db},V_0}(\ub_{\Db},v_0,\btheta_{\Db}) /\partial
\theta_{jg,0}\partial v_0^{k}$, $\partial^{k+1}
\log c_{\Ub_{\Db},V_0}(\ub_{\Db},v_0,\btheta_{\Db}) /\partial
\theta_{jg}\partial v_0^{k}$, $\partial^{k+1}\log
c_{\Ub_{g},V_g;V_0}(\ub_g,v_g;v_0,{\btheta_{g,d_g}})/
\partial\theta_{jg}\partial v_g^{k}$,\\
{$\partial^{k+1}\log
c_{\Ub_{g},V_g;V_0}(\ub_g,v_g;v_0,{\btheta_{g,d_g}})/
\partial\theta_{jg,0}\partial v_g^{k}$}, for $k\in\{0,1,2\}$.
Also, assume the partial derivatives $\partial^{k+1}
\log c_{\Ub_{\Db},V_0}(\ub_{\Db},v_0,\btheta_{\Db}) /\partial
\theta_{jg,0}\partial v_0^{k}$,
$\partial^{k+1}
\log c_{\Ub_{\Db},V_0}(\ub_{\Db},v_0,\btheta_{\Db}) /\partial
\theta_{jg}\partial v_0^{k}$, $\partial^{k+1}\log
c_{\Ub_{g},V_g;V_0}(\ub_g,v_g;v_0,{\btheta_{g,d_g}})/
\partial\theta_{jg}\partial v_g^{k}$,
{$\partial^{k+1}\log
c_{\Ub_{g},V_g;V_0}(\ub_g,v_g;v_0,{\btheta_{g,d_g}})/
\partial\theta_{jg,0}\partial v_g^{k}$} are uniformly bounded for $k\in\{0,1,2\}$
in $\bar{\bm{B}}(\btheta_{\Db},\rho)$.
Then there exists constants $B_{\Db}$ and $B_{\Db}^{*}$ that are bounded
as $d_g\to\infty$ for all $g$, such that:
\begin{align*}
   \|\vtil_{0\Db}(\hat{\btheta}_{\Db})-\vtil_{0\Db}(\btheta_{\Db})
  \|&\leq B_{\Db}\|\hat{\btheta}_{\Db}-\btheta_{\Db}\|^{*}  \\
   \|\vtil_{g\Db}({\hat{\btheta}_{\Db}})-\vtil_{g\Db}({\btheta_{\Db}})&
 \|\leq B_{\Db}^{*}\|{\hat{\btheta}_{g,d_g}-\btheta_{g,d_g}}\|^{*}, 
  \quad g\in\{1,\ldots,G\},
\end{align*}
where
$\|\hat{\btheta}_{\Db}-\btheta_{\Db}\|^{*}:=\sqrt{(2D)^{-1}\sum_{g=1}^{G}
 \sum_{j=1}^{d_g}\|\hat{\btheta}_{jg}-\btheta_{jg}\|^{2}_{2}+\|
 \hat{\btheta}_{jg,0}-\btheta_{jg,0}\|^{2}_{2}}$, and\\
\quad ${\|\hat{\btheta}_{g,d_g}-\btheta_{g,d_g}\|^{*}:=
\sqrt{(2{d_g})^{-1}\sum_{j=1}^{d_g}\|\hat{\btheta}_{jg}-
\btheta_{jg}\|_{2}^{2}+\|\hat{\btheta}_{jg,0}-
\btheta_{jg,0}\|_{2}^{2}}}.$
\end{lemma}

\subsection{Estimation of parameters in blocks}
\label{sec-block_method}

Since the general Gaussian factor model is non-identifiable in terms of
orthogonal rotation of the loading matrix, for a model with two or more
factors, consistency of proxy variables requires a structured loading
matrix such as
that of the bi-factor model or oblique factor model. 

In order to have consistent estimates of parameters, an assumption is
needed on the behavior as the number of variables increase to $\infty$.
A realistic assumption is that
the observed variables (or their correlations, partial correlations,
linking copulas) are sampled from
a super-population.
Then block estimation of parameters is possible, with a finite number
of parameters in each block. In
the structured factor copula models with known parametric family for
each linking copula, the parameters are estimated via sequential maximum
likelihood
with variance of order $O(1/N)$ for sample size $N$. The next lemma
summarizes the procedures for block estimation in the Gaussian and copula
factor models. We unify the notation of observed variables to be $X$;
in Gaussian models, $\bm{X}=\Zb$, and in copula models, $\bm{X}=\Ub$.
Assume that the data from a random sample with $D$ variables
are $\{\bm{X}_{i}: i=1,\ldots,N\}$. 

The idea behind block estimation is that, with the super-population  
assumption, the factor models are closed under margins
(the same latent variables apply to different margins), and parameters 
can be estimated from appropriate subsets or blocks.
Standard maximum likelihood theory for a finite number of parameters
can be applied, and there is no need to develop theory for simultaneous
estimates of all parameters with the number of parameters increasing to
$\infty$.

\begin{lemma}{(Block estimation procedure for 1-factor, bi-factor and oblique factor model).}
\label{block_estimation}
The observed variables are split into several blocks and
each block is a {marginal} factor model linking to the same latent variables.
Estimates of parameters in different blocks are concatenated.
In the Gaussian case, the maximum likelihood (ML) estimates of factor loadings
are unique up to signs, the signs of the
estimates in each block {can} be adjusted appropriately.
\begin{enumerate}
\item (1-factor model): For $j\in{1,2,\ldots,D}$, split
$D$ variables into $K$ blocks of approximate size $B>5$ in a sequential
way, 
the partition
$\mathcal{B}=\{\mathcal{B}_{1},\mathcal{B}_{2},\ldots,\mathcal{B}_{K}\}$
where the cardinality of $\mathcal{B}_{k}$ is $B_k\approx B$ for $k\in\{1,\ldots,K\}$ .
This leads to $K$ marginal
1-factor models with the same latent variable. 
For the convenience of determining the signs of estimated parameters (or positive or negative dependence in the linking copulas), add the first variable $X_1$ in block $\mathcal{B}_2,\ldots,\mathcal{B}_{K}$. 
The estimates of the parameter associated with $X_1$ can be averaged over the
blocks. The estimated parameters are
$(\bm{\hat{\theta}}_{\mathcal{B}_1},\bm{\hat{\theta}}_{\mathcal{B}_2},
\cdots,\bm{\hat{\theta}}_{\mathcal{B}_K})$
after adjusting the signs of estimated parameters in each block.

\item (Oblique factor model): Under the assumption of oblique
factor model, there are $G$ groups with $d_{g}$ dependent
variables in the $g$th group. For each group, split $d_{g}$ variables
into $K$ blocks of approximate size $d_{g}^{(k)}$ dependent variables for 
$k\in\{1,\ldots,K\}$.
Keep the ratio of size of $G$ groups invariant
in each block when splitting, that is ${d_{g}^{(k)}/d_{h}^{(k)}}\approx d_{g}/d_{h}$ for
$g\neq h$, $g,h\in \{1,2,\ldots,G\}$ in block $k$. The partition gives
$K$ blocks, for block $k$, {$\{\bm{X}_{ij\in \mathcal{B}_{k},}:j\in
\{\mathcal{B}_k^{(1)},\ldots,\mathcal{B}_k^{(g)},\ldots,\mathcal{B}_{k}^{(G)}\};\quad i\in\{1,2,\ldots,N\}\}$}
where $\mathcal{B}_{k}^{(g)}$ denotes the $k$th block in the $g$th group. 
For the convenience of determining the signs (or positive or negative dependence in the linking copulas), add $G$ auxiliary variables which are the first variable in $G$ groups for block $\mathcal{B}_1$ to groups in blocks $\mathcal{B}_2,\ldots,\mathcal{B}_{K}$. 
Suppose the {estimates} of parameters involving the variables {in block $k$ of each group} are
$\bm{\hat{\theta}}_{\mathcal{B}_{k}}^{(1)},\bm{\hat{\theta}}_
{\mathcal{B}_{k}}^{(2)},\ldots,
\bm{\hat{\theta}}_{\mathcal{B}_{k}}^{(G)}$, for $k\in\{1,\ldots,K\}$.
{For parameters that are estimated over different blocks, 
such as $\sigb_{\Wb}$,
an average could be taken over the different blocks.}

\item (Bi-factor model) The block method to estimate parameters in
bi-factor model is similar as that used in oblique factor model.
The auxiliary variables that help to determine the signs of parameters are now the first variable $X_1$ and $G-1$ variables which are the first variable in group $g$, $g\in\{2,\ldots,G\}$ for block $\mathcal{B}_1$. The variables are added to $G$ groups in blocks $\mathcal{B}_2,\ldots,\mathcal{B}_{K}$. 
\end{enumerate}
\end{lemma}

\begin{lemma}
{(Asymptotic properties of estimated parameters in 1-factor model, bi-factor model, oblique factor model).}
Suppose there is a sample of size $N$.
For the 1-factor model in  \eqref{eq-1factor-copula-pdf} and \eqref{eq-1factor-gauss}, bi-factor model in \eqref{eq-bifactor-copula-pdf} and \eqref{eq-bifactor-gauss} and oblique factor
model in \eqref{eq-oblique-copula-pdf} and \eqref{eq-oblique-gauss}, for any fixed dimension $D$, let $\bm{\theta}_{D}$ be the parameter
vector in the factor models and $\hat{\bm\theta}_{D}$ be the corresponding
estimates of the parameters, using the block estimation method. Suppose parameters from the linking copulas
behave like a sample from a super-population (bounded parameter space, bounded away from comonotonicity/countermonotonicity), then 
\begin{equation*}
{\hat{\bm{\theta}}}_{j}-{\bm{\theta}}_{j}=O_{p}(N^{-1/2}),\quad { \text{uniformly for}
\ j\in\{1,\ldots,D\}.}
\end{equation*}
\end{lemma}

\begin{proof}
For all factor models in Section \ref{sec-factor}, the block method in Lemma
\ref{block_estimation} gives a partition of variables into several blocks.
In each block, they are {marginal} factor models.
In the Gaussian models,
the results of maximum likelihood estimation 
(\cite{anderson1988asymptotic}) could be applied in each block. 
In the copula factor models, asymptotic maximum likelihood theory could be
applied under standard regularity conditions. The {super-population}
and other regularity assumptions imply that
the expected Fisher information matrices (and standard errors) can be
uniformly bounded over different block sizes.
\end{proof}

Combined the previous results in this section, we
show the consistency of proxy variables with $N,D\to\infty$. 
Due to the consistency of proxy variable with the known parameters, 
with the Lipschitz inequalities, the consistency still hold when the parameters are estimated, as both $N,D\to\infty$. The results in Gaussian and factor copula models are similar, so the results are only stated in the copula case. 

\begin{theorem}{(Consistency of proxies in 1-factor, bi-factor and oblique factor copula models).}
Suppose the 1-factor, bi-factor and oblique factor copula models satisfy the \ref{bivariate} and the regularity conditions in Theorem \ref{consistent_proxy_1fact}, 
Theorem \ref{consistency_proxy_bifct}, 
and Corollary \ref{consistency_proxy_obl} respectively.
Let the parameters $\btheta$ be
$(\btheta_1,\btheta_2,\ldots,\btheta_D,\ldots)$, suppose
$\hat{\btheta}_{j}$ is a estimate  of $\btheta_{j}$. Assume
the factor models are identifiable with respect to parameters, and
$\|\hat{\btheta}_{j}-\btheta_j\|=O_{p}(N^{-1/2})$ for all $j$. Let $\btheta_\Db$
be $\btheta$ truncated to the first $\Db$ random variables $u_{j}$'s,
then as $\Db\to\infty$, {the following hold}. 

(1) For 1-factor copula model, let $(U_{i1},U_{i2},\ldots,U_{iD},\ldots,V_i)$ be a random infinite sequence for $i\in\{1,\ldots,N\}$. 
With $\Vtil_{D}(\hat{\btheta}_{D})$ being the proxy random variable, 
   $ \|\Vtil_{iD}(\hat{\btheta}_{D})-V_{i}\|=o_{p}(1)$.

(2) For bi-factor copula model with $G$ groups, let
$(\Ub_{i1}^{T},\ldots,\Ub_{ig}^{T},\ldots,\Ub_{iG}^{T},V_{i0},V_{i1},\ldots,V_{iG})$
be a random infinite sequence for $i\in\{1,\ldots,N\}$, where
$\Ub_{ig}^{T}=(U_{i1g},U_{i2g},\ldots,U_{id_gg} )$.
With $\Vtil_{0\Db}(\hat{\btheta}_{\Db}),
\Vtil_{1\Db}(\hat{\btheta}_{\Db}),\ldots, \Vtil_{G\Db}(\hat{\btheta}_{\Db})$ 
for the proxy variables, then
   $\|\Vtil_{ig,\Db}(\hat{\btheta}_{\Db})-V_{ig}\|= 
 o_{p}(1)$ for $g\in\{0,1,\ldots,G\}$.

(3) For oblique factor model with $G$ groups, let
$(\Ub_{i1}^{T},\ldots,\Ub_{ig}^{T},\ldots,\Ub_{iG}^{T},V_{i1},\ldots,V_{iG})$
be a random infinite sequence for $i\in\{1,\ldots,N\}$, where
$\Ub_{ig}^{T}=(U_{i1g},U_{i2g},\ldots,U_{id_gg} )$.
With $\Vtil_{1\Db}(\hat{\btheta}_{\Db}),\ldots,$ 
$\Vtil_{G\Db}(\hat{\btheta}_{\Db})$ for the proxy  variables, then
   $\|\Vtil_{ig,\Db}(\hat{\btheta}_{\Db})-V_{ig}\|=
o_{p}(1)$, $g\in\{1,\ldots,G\}$.
\end{theorem}

\begin{proof}
(1) For 1-factor copula model, based on triangle inequality,
\begin{align*}
    \|\Vtil_{iD}(\hat{\btheta}_{D})-V_{i}\| &=\|\Vtil_{iD}
  (\hat{\btheta}_{D})-\Vtil_{iD}({\btheta}_{D})+\Vtil_{iD}
  ({\btheta}_{D})-V_{i}\|\\
   & \leq \|\Vtil_{iD}(\hat{\btheta}_{D})-\Vtil_{iD}({\btheta}_{D})\|+
   \|\Vtil_{iD}({\btheta}_{D})-V_{i}\|
\end{align*}
For the first term in the right-hand side, Lemma \ref{lipschitz} implies that
it has the same order as
$O_{p}(\sqrt{D^{-1}\sum_{j=1}^{D}\|\hat{\btheta}_j-\btheta_j\|^2})=O_{p}(N^{-1/2})$.
By Theorem \ref{consistent_proxy_1fact}, the second term is of order
${O_{p}(D^{-1/2})}$ , then $\|\Vtil_{iD}(\hat{\btheta}_{D})-V_{i}\|=o_p(1)$ 
as $D,N\to\infty$. 
For cases (2) and (3), the same proof technique can be applied in bi-factor
and oblique factor model and we omit details here.
\end{proof}

\begin{remark}
In the proof for consistency of the proxy variables, we assume univariate margins are known. For Gaussian factor models with margins, we assume ${\mu}$'s and $\sigma$'s are known or estimated before
transforming to standard normal. 
For Gaussian factor dependence models and non-Gaussian margins, we assume univariate CDFs are known or have been
estimated before transforming to standard normal.  For factor copulas models, we assume univariate CDFs are known or have been estimated before transforming to U(0,1).  In practice, proxies are estimated
after the estimation of univariate margins, and there is one more source of variability beyond what we studied in this paper.  But consistency and convergence rates are not affected because univariate distributions can be estimated well with a large sample size.
\end{remark}

\section{Sequential estimation for parametric factor copula models}
\label{sec-sequential}

In this section, sequential methods are suggested for estimating the
latent variables and the parameters of the linking copulas,
allowing for choice among several candidate families for each
observed variable. Preliminary diagnostic plots can help to check for
deviations for the Gaussian copula  in terms of tail dependence or tail
asymmetry (\cite{krupskii2013factor}; Chapter 1 of \cite{Joe2014}).

For high-dimensional multivariate data for which initial data analysis
and the correlation matrix of normal scores
suggest a copula dependence structure of 1-factor, bi-factor or oblique
factor, a sequential procedure is presented to estimate the latent variables
with proxies,
decide on suitable families of linking copulas, and estimate parameters of
the linking copulas without numerical integration. Suppose the parametric
linking copula families are not known or specified in advance
(the situation in practice), the sequential method 
starts with unweighted averages estimates in  \cite{krupskii2022approximate} or factor scores  computed from an
estimated loading matrix after observed variables are transformed to
have $N(0,1)$ margins. Then, the ``conditional expectation" proxies
are constructed and are used to estimate the parameters by optimizing
the approximate (complete) log-likelihood with the latent variables
assumed observe at the values of the proxy variables.
The copula density which includes latent variables does not require
the integrals in Section \ref{sec-factor}.
More details are illustrated below. 

{Suppose there is sample of size $N$ from the model, in the 1-factor model, we denote the $i$-th sample as $\ub_{i}=(u_{i1},\ldots,u_{iD})$ and in the bi-factor or oblique factor, we denote the samples as $\ub_{i}=(\ub_{i,1}^{T},\ldots,\ub_{i,G}^{T})$
(the dependence on $d_g$ in $\ub_{ig,d_g}$ is suppressed for simplicity.) }

\indent\textbf{1-factor copula model}.
If the latent variable is assumed observed, then the complete log-likelihood is
{\begin{equation}
\label{approx_1_fact}
 \sum_{i=1}^{N}\log c_{U_{1:D},V}(u_{i1},\ldots,u_{iD},v_i;\bm{\theta}_{D})
 = \sum_{i=1}^{N}\sum_{j=1}^{D}\log {c_{jV}}(u_{ij},v_i;\bm{\theta}_{j}).
\end{equation}
}
\begin{itemize}
\item Stage 1: Define the ``unweighted average" proxy variable as
$U_0=P_{D}^{U}(D^{-1}\sum_{j=1}^{D}U_j)$, where $P_{D}^{U}$ is the cdf of 
$\bar{U}_{D}:=D^{-1}\sum_{j=1}^{D}U_j$.
With enough dependence, $\bar{U}_{D}$ does not converge in probability
to a constant.
{For each sample $i$, $\bar{u}_{i}=D^{-1}\sum_{j=1}^{D}u_{ij}$ and $u_{i,0}={[\text{rank}(\bar{u}_{i})-0.5]}/{N}$;
$\text{rank}(\bar{u}_i)$ is defined as the rank of $\bar{u}_{i}$ based on $\bar{u}_1,\ldots,\bar{u}_{N}$.}
Substitute $v_i=u_{i,0}$ in the log-likelihood
\eqref{approx_1_fact}, and
obtain the first-stage estimates of the parameters in $\bm{\theta}$ from the approximate log-likelihood.
This is the method of \cite{krupskii2022approximate}.

\item Stage 2: Construct the conditional expectations proxies based
on \eqref{eq-1factor-proxy} with the first-stage estimated parameters of
$\bm{\theta}$.  One-dimensional Gauss-Legendre quadrature can be used.
Denote the proxies as $\widetilde{U}_{0}$. {Substitute $v_i=\widetilde{u}_{i,0}$}
in the log-likelihood \eqref{approx_1_fact}, obtain the
second-stage estimates of the parameters from the approximate log-likelihood.
\end{itemize}

\textbf{Bi-factor copula model}.
If the latent variables are assumed observed, then the complete
log-likelihood can be expressed as
\begin{align}
\label{approx_bi_fact}
 & \sum_{i=1}^{N}\log c_{\Ub_{1:D},V_0,V_{g}}{(\ub_{i,1}^{T},\ldots,\ub_{i,G}^{T},v_{i,0},v_{i,g};{\bm{\theta}}}) =
 \sum_{i=1}^{N} \sum_{g=1}^{G}
  \sum_{i=1}^{d_{g}}\log c_{U_{jg},V_{0}}(u_{i,jg},v_{i,0};\btheta_{jg,0})\notag\\
  & \quad+
  \log c_{U_{jg},V_{g};V_{0}}(C_{U_{jg}|V_{0}}(
  u_{i,jg}|v_{i,0}),v_{i,g};\btheta_{jg}).
\end{align}
Suppose the bi-factor structure is known, i.e., the number of groups and
the number of variables in each group, estimation can be performed in
two stages.  
\begin{itemize}
\item Stage 1:  Assume the variables are monotonically related
and that the Gaussian copula is reasonable as a first-order model. Convert data into normal scores and fit a Gaussian model with bi-factor structure. Compute the factor scores in
\eqref{eq-bifactor-fs-W0} and \eqref{eq-bifactor-fs-Wg}, and denote
as $\wtil_{i0},\wtil_{i1},\ldots, \wtil_{iG}$. The first-stage proxy
variable are defined as $\Vtil_{0}^{(1)}=P_{D,0}^{U}(\Wtil_{0})$,
$\Vtil_{g}^{(1)}=P_{d_g,g}^{U}(\Wtil_{g})$, $g\in\{1,\ldots,G\}$; where $P_{D,0}^{U}$ is the cdf of $\Wtil_{0}$, and $P_{d_g,g}^{U}$ is the cdf of $\Wtil_{g}$, $g\in\{1,\ldots,G\}$.
Letting $v_{i,0}=\vtil_{i,0}^{(1)}$,
$v_{i,g}=\vtil_{i,g}^{(1)}$ in log-likelihood
\eqref{approx_bi_fact}, 
obtain the first-stage estimates of the parameters from the approximate log-likelihood. 

\item Stage 2: Construct the conditional expectation proxies based on
equation \eqref{eq-bifactor-gproxy} and \eqref{eq-bifactor-lproxy} with first-stage
estimates plugged in. Nested 1-dimensional Gauss-Legendre quadrature can be used. Denote the conditional expectation proxies as
$\Vtil_{0}^{(2)}$,  $\Vtil_{g}^{(2)}$, $g\in\{1,\ldots,G\}$. Letting
$v_{i,0}=\vtil_{i,0}^{(2)}$, $v_{i,g}=\vtil_{i,g}^{(2)}$
in log-likelihood \eqref{approx_bi_fact}, obtain the
second-stage estimates of the parameters from the approximate log-likelihood.
\end{itemize}

\textbf{Oblique factor copula model}.
Similar to the bi-factor copula model, assume the group structure of
the model is known. If the latent variables are assumed observed, the complete
log-likelihood is
\begin{align}
\label{approx_obl} 
  &  { \sum_{i=1}^{N}\log c_{\text{oblique},\Ub_{1:D},V_{1:G}}({\ub}_{i,1}^{T},
 \ldots,{\ub}^{T}_{i,G},v_{i,1},\ldots,v_{i,G};\bm{\theta})
 = \sum_{i=1}^{N}\sum_{g=1}^{G}\sum_{j=1}^{d_g}\log c_{{{U_{jg},V_g}}}(u_{i,jg},v_{i,g};
 \btheta_{jg})} \notag \\
 & \quad+  \log c_{\bm{V}}(v_{i,1},\ldots,v_{i,G}; \btheta_{\bm{V}}),
\end{align}
where
$\btheta=(\btheta_{d_1,1}^{T},\ldots,\btheta_{d_G,G}^{T},\btheta_{V}^{T})$,
and $c_{\Vb}$ is the copula density of the latent variables.
\begin{itemize}
\item Stage 1: For $g\in\{1,\ldots,G\}$, let
$\bar{U}_{g}=P^{U}_{d_g,g}(d_g^{-1}\sum_{j=1}^{d_g}U_{jg})$, where $P^{U}_{d_g,g}$ is the cdf of $\bar{U}_{g}:=d_{g}^{-1}\sum_{j=1}^{d_g}U_{jg}$. {For each sample, $\bar{u}_{i,g}=d_{g}^{-1}\sum_{j=1}^{d_g}u_{i,jg}$ and $u_{i,g}=[\text{rank}(\bar{u}_{i,g})-0.5)]/N$, where $\text{rank}(\bar{u}_{i,g})$ is defined as the rank of $\bar{u}_{i,g}$ based on $\bar{u}_{1,g},\ldots,\bar{u}_{N,g}$.}
Let $v_{i,g}={u}_{i,g}$ in log-likelihood
\eqref{approx_obl} to get the first-stage estimates of the parameters from the approximate log-likelihood.
This is the method of \cite{krupskii2022approximate}.

\item Stage 2: Construct the conditional expectations proxies based on
\eqref{eq-oblique-proxy} with the first-stage estimated parameters. {This requires 1-dimensional numerical integration.}
Denote the proxies as $\widetilde{U}_{g}$. Substitute $v_{i,g}=\util_{i,g}$
in the log-likelihood \eqref{approx_obl}, and obtain the
second-stage estimates of the parameters.
\end{itemize}

For 1-factor and oblique factor models, unweighted averages can be
consistent under some mild conditions \cite{krupskii2022approximate}, but the above methods based
on conditional expectations perform better from the simulation
results shown in the next section. 
{The estimation of proxies and copula parameters could be iterated
further if desired stage 2 estimates differ a lot from stage 1
estimates.}

For optimizing the above approximate log-likelihoods
using proxies for the latent variables, we adopt a modified Newton-Raphson
algorithm with analytic derivatives; see \cite{krupskii2015structured}
for details of the numerical implementation.

\section{Simulation experiments}
\label{sec-simulation}

This section has some simulation results to support and explain ideas in previous sections. In all the following settings, the parameters on the linking copulas are designed to be generated uniformly from a {bounded subset of the parameter space}; this is an example of  sampling from some super-population. 
{Many different scenarios were assessed and some representative
summaries are given in three subsections for the 1-factor, bi-factor and
oblique factor copula models}.

{The sequential approach of Section \ref{sec-sequential}
is compared with the ``exact" method from the implementation
of \cite{krupskii2013factor}\cite{krupskii2015structured} with
R front-end and FORTRAN 90 back-end for minimizing the negative
log-likelihood with a modified Newton-Raphson algorithm; the ``exact"
method is indicated with the superscript $m=0$.  The proxy approach
in \cite{krupskii2022approximate} for 1-factor and oblique factor,
as summarized in Section \ref{sec-sequential} is indicated with the
superscript $m=1$. The stage 2 estimates for all three copula models is
indicated with the superscript $m=2$.
When the linking copula families are assumed known, differences of copula
parameter estimates are summarized in the {Kendall's} tau scale.
If linking copula families are decided based on a few parametric choices
that cover a range of tail asymmetry and strength of dependence in
joint tails, additional summaries are based on the tail-weighted 
dependence measures defined in \citep{lee2018tail}. This is because
different bivariate copula families can have members that are similar
in tail properties.
}

\subsection{One-factor copula model}

Two settings are summarized in Table \ref{tab:one_fact_simu} to illustrate the performance of the sequential approach. 
The sample size is $N$ and there are $D$ variables with bivariate linking copulas $C_{jV}(\cdot,\cdot;\btheta_j)$, $j\in\{1,\cdots,D\}$, to the latent variable. 
The parameters $\bm{\theta}=(\btheta_{1},\ldots,\btheta_{D})$ are
independent $U(\theta_{L},\theta_{U})$, where, $\theta_{L}$ and $\theta_{U}$ are chosen so that the Kendall's taus of the bivariate copulas range between 0.4 and 0.8 for moderate to strong dependence.  
The simulation size is 1000.

\begin{table}[H]
\centering
\begin{tabular}{cccccc}
\toprule
 & N & D & linking families & $[\theta_{L},\theta_{U}]$\\
\hline
setting1 & 500 & (20,40,60,80) & Frank & [4.2.18.5]\\
setting2 & 500 & (30,45,60,90) & Gumbel, t, Frank & $\{(1.67,5), (4.2,18.5), (0.59,0.95)\}$ \\
\bottomrule
\end{tabular}
\caption{\footnotesize{Two simulation settings for the 1-factor copula model; in setting2, the number of linking copulas in different families are {approximately $D/3$}, and the $\nu$ parameter of Student-t copulas is fixed at 5. The parameters are chosen to let the Kendall's tau be in [0.4,0.8].}}
\label{tab:one_fact_simu}
\end{table}

In setting1, the main summary is the mean absolute error (MAE) of estimated parameters for three different methods ($m=0,1,2$ as indicated above):
$\hat{\theta}_{\text{MAE}}^{m}=(ND)^{-1}\sum_{i=1}^{N}\sum_{j=1}^{D}|\hat{\theta}_{ij}^{m}-\theta_{ij}|$, 
where $\theta_{ij}$ is the parameter of $C_{j,0}$ generated at the $i$th simulation and $\hat{\theta}_{ij}^{m}$ is the corresponding estimate using the different approaches. 
An additional summary is the MAE of the differences between the estimates obtained from the proxy methods and the exact approach, as well as the differences
of corresponding Kendall's taus (function of the estimated bivariate linking copula): 
\begin{align}
    \hat{\theta}_{|\text{diff}|}^{m}=\frac{1}{ND}\sum_{i=1}^{N}\sum_{j=1}^{D}|\hat{\theta}_{ij}^{m}-\hat{\theta}_{ij}^{0}|, \quad 
    \hat{\tau}_{|\text{diff}|}^{m}=\frac{1}{ND}\sum_{i=1}^{N}\sum_{j=1}^{D}|\hat{\tau}_{ij}^{m}-\hat{\tau}_{ij}^{0}|, \quad {m=1,2}.
  \label{eq-MAE-param}
\end{align}

In setting2, summaries include the averaged differences of the dependence measures between the true and fitted models over the $D$ bivariate linking copulas:
\begin{align}
[\hat{M}^{\text{diff}}]_{\text{mean}}^{m} &=\frac{1}{D}
 \sum_{j=1}^{D}|[\widehat{M}^{\text{model}}]_{j}-[M^{\text{true}}]_{j}|,
  \quad m=1,2,
  \label{eq-MAE-depmeas}
\end{align}
where the measure $M$ can be Kendall's tau, and tail-weighted upper/lower tail dependence as defined in \citep{lee2018tail}. Denote them as $M={\tau},\zeta_{\alpha,U}(20),\zeta_{\alpha,L}(20)$ respectively. 

To compare the proxies as estimated latent variables, a summary is the RMSE of the proxies of the two methods:
\begin{align}
    \hat{v}_{\text{RMSE}}^{m}&=\Bigl\{\frac{1}{KN}\sum_{k=1}^{K}
  \sum_{i=1}^{N}\big(\hat{v}^{m}_{ki}-v_{ki}\big)^2\Bigr\}^{1/2},
  \quad m=1,2,
  \label{eq-RMSE-vhat}
\end{align}
where $v_{ki}$ is the latent variable for the $i$th {observation vector} in the $k$th simulation.

\begin{table}[H]
\centering
\begin{tabular}{c|ccc|cc|cc|cc}
\toprule
$D$ & $\hat{\theta}_{\text{MAE}}^{0}$&
$\hat{\theta}_{\text{MAE}}^{1}$&
$\hat{\theta}_{\text{MAE}}^{2}$&
{$\hat{\theta}_{|\text{diff}|}^{1}$} & 
     $\hat{\theta}_{|\text{diff}|}^{2}$&
{$|\hat{\tau}^{\text{diff}}|_{\text{mean}}^{1}$} & 
     $|\hat{\tau}^{\text{diff}}|_{\text{mean}}^{2}$& 
$\hat{v}_{\text{RMSE}}^{m=1}$& $\hat{v}_{\text{RMSE}}^{m=2}$\\
\hline
20 & 0.493 & 0.680 & 0.750 & 0.545 & 
0.578 & 0.010 & 0.008 & 0.041 & 0.032\\
40 & 0.474 & 0.582 & 0.552 & 0.374 & 
0.283  & 0.006 & 0.004 & 0.032 & 0.025\\
60 & 0.472 & 
0.565 & 0.505 & 0.320 & 0.184 & 0.005 & 0.002 & 0.028 & 0.023\\
80 & 0.471 & 0.544 & 0.485 & 0.270 & 
0.138 & 0.004 & 0.002 & 0.026 & 0.016\\
\bottomrule
\end{tabular}
\caption{\footnotesize{1-factor copula models with all Frank linking copulas; simulation size 1000, sample size $N=500$,  
$\bm{\theta}$ uniform in $(\bm{\theta}_{L},\bm{\theta}_{U})$ as
specified in Table \ref{tab:one_fact_simu}.
{Summaries from \eqref{eq-MAE-param}, \eqref{eq-MAE-depmeas} and \eqref{eq-RMSE-vhat}}
for 3 approaches --- superscript $m=0$: exact; superscript $m=1$: unweighted average proxy; superscript $m=2$: sequential.}}
\label{tab: 1-fact-comparision}
\end{table}

\begin{table}[H]
\centering
\begin{tabular}{c|cccccccc}
\toprule
$D$  &  $[\hat{\tau}^{\text{diff}}]_{\text{mean}}^{m}$  &  
$[\hat{\zeta}_{\alpha, U}^{\text{diff}}]_{\text{mean}}^{m}$  &   
$[\hat{\zeta}_{\alpha, L}^{\text{diff}}]_{\text{mean}}^{m}$  &   
$\hat{v}_{\text{RMSE}}^{m=1}$  &   $\hat{v}_{\text{RMSE}}^{m=2}$\\
\hline
30 & 0.017/0.015 & 0.028/0.025 & 0.026/0.019 & 0.041 & 0.027\\
45 & 0.016/0.014 & 0.024/0.021 & 0.021/0.016 & 0.035 & 0.022\\
60 & 0.015/0.013 & 0.022/0.019 & 0.019/0.014 & 0.032 & 0.019\\
90 & 0.014/0.012 & 0.020/0.017 & 0.017/0.014 & 0.028 & 0.017\\
\bottomrule
\end{tabular}
\caption{\footnotesize{1-factor copula model with linking copulas from Gumbel, t and Frank families; Simulation size 1000, sample size $N=500$. In each simulation $\btheta$ is uniform in $(\btheta_L,\btheta_{U})$ {as specified in Table
\ref{tab:one_fact_simu}.}
{The summaries are for \eqref{eq-MAE-depmeas} and \eqref{eq-RMSE-vhat}}
and  are shown in order $m=1/m=2$ respectively in columns 2 to 4. }}
\label{tab:1-fact-comp-2}
\end{table}

From Table \ref{tab: 1-fact-comparision}, the two proxy approaches can give accurate parameter estimates comparable to those obtained from the exact likelihood when $D\geq 40$. The sequential approach performs better than the ``unweighted average" approach.
The differences in dependence measures between the two proxy approaches and the exact approach decrease with the increasing dimension. The sequential approach gives the parameter estimates closer to the exact approach than that in \cite{krupskii2022approximate}. In addition, the conditional expectation proxies are closer to the true realized latent variables.
Similar observations can be seen in Table \ref{tab:1-fact-comp-2}. In setting2, the proxy approach can identify the correct copula families in most cases, though sometimes the method selects BB1 copulas with similar tail behavior to the true linking Gumbel copulas.

\subsection{Bi-factor copula model}

Two settings are summarized in Table \ref{tab:bi_fact_simu} to illustrate the sequential approach. The sample size is $N$ and there are $D$ variables and  $2D$ linking copulas. The number of groups $G=3$ and the size of each group is
approximately $D/3$. The parameters of the $D$ copulas linking the observed variables and the global latent variable are generated uniformly in $(\theta_{L},\theta_{U})$ so that there is a wide range for the dependence between the observed variables and the global latent variable. 
For the $D$ bivariate copulas {for conditional dependence}, the parameters are generated uniformly from $(\theta_{LL},\theta_{UU})$ so that the within-group dependence is strong. Also, the parameter setting ensures the condition number (in Remark 1) of the $\Qb_{\Db}$, obtained from fitted bi-factor Gaussian factor structure on the data transformed to $N(0,1)$ scales, is small enough
for a reasonable convergence rate.
In both setting, the simulation size is 1000.

\begin{table}[H]
\centering
\begin{tabular}{cccccc}
\toprule
 & $N$ & $D$ & linking families & 
$[\btheta_{L},\btheta_{U}]$ ($\tau$) & 
${[\btheta_{LL},\btheta_{UU}]}$ ($\tau$)\\
\toprule
setting1 & 1200 & (30,60,90,120) & Frank/Frank &  [1.87,8] (0.2,0.6) & 
[4.2,11.5] (0.4,0.7) \\
setting2 &  2000 & (30,60,90) & BB1/Frank & $[0.3,1]\times[1.1,2.5]$ (0.2,0.7) &
[8.5,18.5] (0.6,0.8)\\
\bottomrule
\end{tabular}
\caption{\footnotesize{Two simulation settings for the bi-factor copula model; In setting2, the $D$ global linking copulas are in the BB1 family and the $D$ local linking copulas are in the Frank family. For the comparison in setting2, $D=120$ would take too much computational time for the exact approach. The range of Kendall's tau corresponds to the range of parameters are included after the parameter interval. The condition number of $\Qb$ matrix is around 50 in setting1 and 70 in setting2.} }
\label{tab:bi_fact_simu}
\end{table}

As in the previous subsection, summaries are MAE of estimated parameters in the global and local linking copulas. 
Also summarized are differences in estimated parameters and corresponding Kendall's tau between the proxy method and the exact method. 
For setting2, the sequential approach is applied in two cases: (a) assuming the linking copula families are known; (b) assuming the linking copula families are to be decided. 
In setting1, the linking copula families are assumed known. 
In setting2 with case (b), summaries are
{as in \eqref{eq-MAE-depmeas}}.
The simulation results are summarized in Table \ref{tab:compare-bi-factor} and Table \ref{tab:compare-bi-factor-2}.

\begin{table}[H]
\centering
\begin{tabular}{*{12}{c}}
\toprule
$D$ & \multicolumn{3}{c}{Global linking copulas} & 
\multicolumn{3}{c}{Local linking copulas} & 
\multicolumn{3}{c}{$\text{RMSE}_{\text{proxy}}$}\\
\cmidrule(lr){2-4}\cmidrule(lr){5-7}\cmidrule(lr){8-9}
\morecmidrules
Frank & $\hat{\theta}_{\text{glob:MAE}}^{m}$ & $\hat{\theta}_{\text{glob}:|\text{diff}|}^{m=2}$ & 
$\hat{\tau}_{\text{glob}:|\text{diff}|}^{m=2}$ & 
$\hat{\theta}_{\text{loc:MAE}}^{m}$ & 
$\hat{\theta}_{\text{loc}:|\text{diff}|}^{m=2}$ & 
$\hat{\tau}_{\text{loc}:|\text{diff}|}^{m=2}$ & 
$\hat{v}_{0}/\hat{v}_{g}$\\
\hline
30 & 0.160/0.352 & 0.306 & 0.020 & 0.260/1.000 & 0.959 & 0.029 & 0.073/0.100\\
60 & 0.145/0.256 & 0.201 & 0.012 & 0.239/0.461 & 0.377 & 0.012 & 0.054/0.073\\
90 & 0.139/0.234 & 0.180 & 0.010 & 0.230/0.320 & 0.201 & 0.007 & 0.045/0.061\\
120 & 0.135/0.233 & 0.183 & 0.010 & 0.227/0.283 & 0.159 & 0.006 & 0.040/0.053\\
\bottomrule
\end{tabular}
\caption{\footnotesize{Bi-factor copula model with all linking copulas in the Frank family; Simulation size 1000, sample size $N=1200$, number of variables in each group equal and set to be $d_g= 10,20,30,40$, $g=1,2,3$. 
The parameters are generated
{as specified in Table \ref{tab:bi_fact_simu}}. For $D$ global/local linking copulas, summaries of
$\hat{\theta}_{\text{MAE}},\hat{\theta}_{\text{diff}},\hat{\tau}_{\text{diff}}$
for approaches superscript $m=0$: exact; superscript $m=2$: sequential are shown; For MAE, the results are shown for m=0/m=2 respectively.} }
\label{tab:compare-bi-factor}
\end{table}

From Table \ref{tab:compare-bi-factor}, the proxy method can give parameter estimates close to that of exact approach when $d_{g}\geq 20$. The differences in the estimates between the proxy and exact methods decrease as the dimension becomes large. Also, the sequential conditional expectation proxies are getting closer to the realizations of latent variables as $d_g$ increases. The global latent variables can be estimated more accurately than the local ones since all the observed variables are used for estimation. A similar observation can be found in Table \ref{tab:compare-bi-factor-2} in the case where copula families are not specified. 
From the results in the second part of Table \ref{tab:compare-bi-factor-2},  the dependence measures of the estimated linking copulas are close to those of the actual linking copulas. The method will {also} select t, Gumbel copula, or survival BB1 copula for the global linking copulas, and the selected copulas have similar tail behaviors to the actual ones. The RMSE of the proxies indicates that the latent variables can be estimated well even though some of the linking copula families are misspecified.
\begin{table}[H]
\centering
\begin{tabular}{*{12}{c}}
\toprule
$D$&\multicolumn{3}{c}{Global linking copulas}&
\multicolumn{3}{c}{Local linking copulas}&
\multicolumn{3}{c}{$\text{RMSE}_{\text{proxy}}$}\\
\cmidrule(lr){2-4}\cmidrule(lr){5-7}\cmidrule(lr){8-9}
\morecmidrules
case (a) & $\hat{\theta}_{\text{glob:MAE}}^{m}$ & 
$\hat{\theta}_{\text{glob}:|\text{diff}|}^{m=2}$ & 
$\hat{\tau}_{\text{glob}:|\text{diff}|}^{m=2}$ & 
$\hat{\theta}_{\text{loc:MAE}}^{m}$ & 
$\hat{\theta}_{\text{loc}:|\text{diff}|}^{m=2}$ & 
$\hat{\tau}_{\text{loc}:|\text{diff}|}^{m=2}$ & 
$\hat{v}_{0}/\hat{v}_{g}$\\
\hline
30 & 0.032/0.145 & 0.142 & 0.013 & 0.377/1.290 & 0.965 & 0.013 & 0.042/0.075\\
60 & 0.030/0.144 & 0.141 & 0.008 & 0.278/0.440 & 0.349 & 0.008 & 0.030/0.059\\
90 & 0.029/0.140 & 0.139 & 0.010 & 0.265/0.416 & 0.358 & 0.009 & 0.027/0.053\\
\midrule
case (b)  & \footnotesize{$[\hat{\tau}^{\text{diff}}]_{\text{mean}}^{m=2}$} & 
\footnotesize{$[\hat{\zeta}_{\alpha,U}^{\text{diff}}]_{\text{mean}}^{m=2}$} & 
\footnotesize{$[\hat{\zeta}_{\alpha,L}^{\text{diff}}]_{\text{mean}}^{m=2}$} & 
\footnotesize{$[\hat{\tau}^{\text{diff}}]_{\text{mean}}^{m}$} & 
\footnotesize{$[\hat{\zeta}_{\alpha,U}^{\text{diff}}]_{\text{mean}}^{m=2}$} & 
\footnotesize{$[\hat{\zeta}_{\alpha,L}^{\text{diff}}]_{\text{mean}}^{m=2}$}
 & $\hat{v}_{0}/\hat{v}_{g}$\\
30 & 0.013 & 0.024 & 0.037 & 0.015 & 0.018 & 0.018 & 0.042/0.075\\
60 & 0.012 & 0.023 & 0.028 & 0.007 & 0.008 & 0.008 & 0.031/0.058\\
90 & 0.013 & 0.025 & 0.022 & 0.008 & 0.009 & 0.009 & 0.027/0.053\\
\bottomrule
\end{tabular}
\caption{\footnotesize{Bi-factor copula model with BB1 global linking copulas and Frank local linking copulas; 
Simulation size 1000, sample size $N=2000$, number of variables in each group equal and set to be $d_g= 10,20,30$, $g=1,2,3$. For $D$ global/local linking
copulas, $\hat{\theta}_{\text{MAE}},\hat{\theta}_{\text{diff}},\hat{\tau}_{\text{diff}}$ for approaches superscript $m=0$: 
exact; superscript $m=2$: sequential are shown. 
For MAE, the results are shown for m=0/m=2 respectively;
In case (b),  $[{\hat\tau}^{\text{diff}}]_{\text{mean}}$, 
$[\hat{\zeta}_{\alpha, U}^{\text{diff}}]_{\text{mean}}$,
$[\hat{\zeta}_{\alpha, L}^{\text{diff}}]_{\text{mean}}$ 
are the averaged differences in the dependence measures between true and fitted models over $D$ global/local linking copulas. 
The results are only shown for the sequential method.
} }
\label{tab:compare-bi-factor-2}
\end{table}

\subsection{Oblique factor model}

A simulation setting consists of $K=1000$ replications of sample size $N=1000$ from a (nested) oblique copula model in \eqref{eq-oblique-copula-pdf} with $G=3$ groups of equal group size $d_{g}$; $d_g=10, 15, 20, 30$. 
The density $c_{V}$ is assumed to have a one-factor structure with Frank linking copula, {because this is the nested copula
in \cite{krupskii2015structured} for which only 2-dimensional quadrature
is needed instead of $G$-dimensional.}
The parameters in $c_V$ are generated uniformly in  $(\theta_{L},\theta_{U})=(3,6)$ such that the Kendall's tau is between 0.3 to 0.5. The bivariate linking copulas in three groups are in the Gumbel, Frank and t families respectively. 
For each group, the parameters for the linking copulas are generated uniformly in $(\theta_{LL},\theta_{UU})$, for Gumbel, BB1 and student-t copulas, with $(\theta_{LL},\theta_{UU})$ in $(1.67,5)$, $(0.25,2)\times (1.5,2.5)$, $(0.59,0.95)$
respectively. The $\nu$ parameter of Student-t copulas is fixed at 5. The Kendall's taus for linking copulas in each group are between 0.4 and 0.8. 
Proxy variables are used to decide on the families for the linking copulas in each group. The simulation results are summarized in the Table \ref{obl_sim2}.

\begin{table}[H]
\centering
\begin{tabular}{*{10}{c}}
\toprule
$D$ & \multicolumn{2}{c}{Global linking copulas} & 
\multicolumn{3}{c}{Local linking copulas} & 
\multicolumn{3}{c}{$\text{RMSE}_{\text{proxy}}$}\\
\cmidrule(lr){2-3}\cmidrule(lr){4-6}
\morecmidrules
 & $\hat{\theta}_{\text{RMSE}}$ & 
$[\hat{\tau}^{\text{diff}}]_{\text{mean}}^{m}$ & 
$[\hat{\tau}^{\text{diff}}]_{\text{mean}}^{m}$ & 
$[\hat{\zeta}_{U,\alpha}^{\text{diff}}]_{\text{mean}}^{m}$ &  
$[\hat{\zeta}_{L,\alpha}^{\text{diff}}]_{\text{mean}}^{m}$\\
\hline
30 & 0.514/0.489 &  0.028/0.025 &  0.022/0.029  & 0.029/0.023 & 
0.024/0.019 &  0.067/0.056 \\
45 & 0.488/0.475 & 0.026/0.024 & 0.016/0.021 & 0.025/0.018 & 0.022/0.016  & 0.056/0.045\\
60 & 0.471/0.460 & 0.025/0.023 & 0.014/0.017 & 
0.022/0.016 &  0.020/0.015 & 0.049/0.039\\
90 & 
0.479/0.478 & 0.025/0.024 & 0.012/0.014 & 0.019/0.015 & 0.018/0.013 & 0.040/0.031\\
\bottomrule
\end{tabular}
\caption{\footnotesize{Oblique factor copula model with $G=3$ groups of equal group size $d_{g}$; $d_g=10, 15, 20, 30$; Sample size $N=1000$;  Simulation size is 1000. 
{Summaries from \eqref{eq-MAE-depmeas} and \eqref{eq-RMSE-vhat}}
for 2 approaches --- superscript $m=1$: unweighted proxy; superscript $m=2$: sequential. The results are shown for $m=1/m=2$ respectively.  
$\hat{\theta}_{\text{RMSE}}^{m}=\sqrt{\frac{1}{NG}\sum_{i=1}^{N}\sum_{g=1}^G|\hat{\theta}_{ig}^{m}-\theta_{ig}|^2}$, where $\theta_{ig}$ is the parameter of the copula linking the $g$th group-specific latent variable with the latent variable at the $i$th simulation.}}
\label{obl_sim2}
\end{table}

From Table \ref{obl_sim2}, the differences in Kendall’s tau, and empirical dependency measures are decreasing as the dimension becomes large. The two proxy methods can perform well when $d_{g}\geq 20$ provided the within-group dependence is strong. From the results in local linking copulas, the unweighted average proxy approach has slightly smaller averaged differences
in Kendall’s tau while the sequential approach has smaller averaged differences in the empirical upper and lower dependence measures, as defined in
\eqref{eq-MAE-depmeas}. In addition, the conditional expectation proxies have smaller
RMSEs and are closer to the true realized latent variables.

\section{Factor models with residual dependence}
\label{sec-residdep}

It is important that we can show that proxy estimates for latent variables can be adequate for some factor models when the sample size is large enough and there are enough observed variables linked
to each latent variable. However, as the number of variables increase, it is unlikely that factor models 
with conditional independence given latent variable continue to
hold exactly. 
\cite{krupskii2022approximate} have a partial study of their simple proxies in the case of weak conditional dependence of observed
variables given the latent variables. This is called
\textit{weak residual dependence}; see also \cite{joe2018parsimonious}
and references therein.

For the proxies in Section \ref{sec-proxies}, we have obtained 
conditions for weak residual dependence for which these proxies
(derived based on assumption of conditional independence)
are still consistent. We indicate a result in this section for the Gaussian 1-factor model.
There are analogous conditions for the 1-factor, bi-factor and
oblique factor copula models.

{With the linear representation as the Gaussian 1-factor model defined in \eqref{eq-1factor-gauss}, the residual dependence indicates that $\epsb_j$ are not independent. 
Let $\bm{\Omega}_{D}$ be the correlation matrix of 
$\epsb_{D}=(\eps_1,\ldots,\eps_{D})^{T}$, $\gammab_{D}=\Psib_{D}\bm{\Omega}_{D}\Psib_{D}$, 
and {$\Psib^2_{D}$} is a diagonal matrix with diagonal entries of $\gammab_{\Db}$. 
The factor scores are defined as $ \wtil_{D}=(\bm{I}+\Ab^{T}_{D}\Psib_{D}^{-2}\Ab_{D})^{-1}\Ab_{D}^{T} \Psib_{D}^{-2}\zb_{D}$.
 
Suppose the maximum eigenvalue of matrix $\bm{\Omega}_{D}$ is
bounded {as $D\to\infty$,}
the model is an approximate factor model
from the definition in \cite{chamberlain1983arbitrage}. 
This assumption is sufficient for the defined proxy to be
asymptotically consistent. An equivalent assumption, which is easier to check, is given  below. Similar assumptions are presented in \cite{bai2016maximum}. 

\begin{assumption}
\label{errors}
Let $\bm{\Omega}_{\Db}=(\omega_{s,t})_{1\leq s,t\leq D}$ be the  correlation matrix of $\epsb_{D}$. Let $S_{D}=\sum_{j=1}^{D}\epsilon_{j}$ and $\epsbar_D=S_D/D$, then
$\e(S_{D}^2)=\sum_{s=1}^{D}\sum_{t=1}^{D}\omega_{s,t}$. Assume
\begin{equation*}
    0<\lim\inf_{D\to\infty}\frac{\e (S_{D}^2)}{D}<\lim\sup_{D\to\infty}
  \frac{\e (S_{D}^2)}{D}<M
\end{equation*}
where $M$ is a positive constant.
\end{assumption}

{The above Assumption implies $\Var(\epsbar_D)=O(D^{-1})$, the same order
as the case of iid.}
Under the Assumption \ref{errors} and the assumption on the loadings in Theorem \ref{1_factor_gaussian}, it is shown in \ref{sec-proof-residdep} that
$\wtil_{D}-w^{0}=O_{p}(D^{-1/2})$ as $D\to\infty$. 
That is, if the residual dependence is weak, the consistency of the proxy variable defined from a slightly misspecified model still holds with the same convergence rate.

\begin{remark}
\label{example}
The assumptions indicate the summation of entries in matrix residual correlation matrix $\bm{\Omega}_{D}$ is of $O(D)$. Let $\omega_{s+}=\sum_{t=1}^{D}\omega_{st}$. If $\omega_{s+}$ is $O(1)$ as $D\to\infty$ for all $s$, then Assumption is satisfied; e.g., $\epsilon$'s are indexed to have ante-dependence of order 1: $0<r_1<\omega_{j,j+1}<r_{2}<1$ for all $j$, and 
{$\omega_{jk}=\prod_{i=j}^{k-1} \omega_{i,i+1}$ for $k-j\ge2$}. 
If $\omega_{s+}$ is $O(D)$ as $D\to\infty$ for all $s$, then Assumption is not satisfied; e.g., $\epsilon_1$ is dominating: $0<r_1<\omega_{1,j}<r_2<1$ for all $j$, and $\omega_{jk}=\omega_{1j}\cdot\omega_{1,k}$.
\end{remark}

In the 1-factor copula model with weak residual dependence, with the same notations and assumptions in Theorem \ref{consistent_proxy_1fact}, there are similar sufficient conditions. 
{If the copula for residual dependence is multivariate Gaussian, a sufficient condition is}:
\begin{equation}
\label{res_tree}
{\sum_{j=1}^{D}\sum_{k=1}^{D}\text{Cor}\bigl(C_{{U_j|V}}(U_j|V=v^0),C_{{U_k|V}}(U_k|V=v^0)\bigr)=O(D)
\quad \forall\,0<v^0<1. }
\end{equation}
Similar ideas extend to residual dependence for bi-factor and oblique factor copulas.

\section{Discussion and further research}
\label{sec-discussion}

This paper proposes the conditional expectation proxies of the latent variables in some factor copula models and shows the consistency of proxy variables under some mild conditions.
For high-dimensional factor copula models with a large sample size (large $N$, large $D$), simulation studies show that the sequential estimation approach can efficiently estimate the latent variables and select the families of linking copulas as well as estimate the copula parameters.

There are other recent methods for factor copula models that use Bayesian computing methods. 
For 1-factor copula model,
\cite{tan2019bayesian} use reversible jump MCMC to select the bivariate copula links during the sampling process and to make inferences of the model parameters and latent variables. 
\cite{nguyen2020variational} utilize a Bayesian variational inference algorithm to make  inferences for structured factor models but they make a strong assumption on the form of posterior distributions.
Compared to their approaches, our inference method is more intuitive and {does not need to fix a factor structure. 
The sequential procedures fit better with the use of Gaussian factor models as a started point
to consider different factor structures that fit the data.}

Our sequential proxy methods improve on the approach in \cite{krupskii2022approximate} for 1-factor and oblique factor models, and can handle bi-factor copula models under some conditions. {The sequential proxy procedures require numerical integration to compute
second-stage proxies but not for maximum likelihood iterations for copula parameters, and hence the computation effort is reduced at lot.}
The simulation studies show the conditional expectation proxies are usually closer to the realized latent variables, leading to more accurate estimates of the parameters than that obtained from the ``unweighted average" proxy approach in \cite{krupskii2022approximate} in the 1-factor or oblique factor models.

Applications of factor copula models making use of the theory
in this paper will be developed separately.
Topics of further research and applications include the following.

(a) If the 1-factor structure is not adequate and group structure of
observed variables cannot be determined from context, then a $p$-factor
structure with varimax rotation can be fit to observed variables in the
normal scores scale to check if an interpretable loading matrix with many
zeros, corresponding to variables in overlapping groups, can be found.
If so, for the factor copula counterpart, the sequential approach for the
bi-factor copula can be extended. If the number $p$ of latent variables is
three of more, the exact copula likelihood would require $p$-dimensional
Gaussian quadrature and we would not be able to compare estimation of
copula parameters via proxies and via the exact likelihood.
However the theory and examples in this paper suggest that the proxy
approach will work if the number of variables linked to each latent variable
is large enough.

{(b) If one latent variable can explain much of the dependence but 
any $p$-factor loading matrix (with $p\ge2$) is not interpretable,
one could consider a 1-factor model with weak or moderate residual
dependence. Starting with a preliminary 1-factor copula with residual
dependence, one can iterate as in Section \ref{sec-sequential}
and get proxies from the conditional expectation of the latent
variable given the observed variables, from which to get better choices
for the bivariate linking copulas to the latent variable.
At most 1-dimensional Gaussian quadrature would be be needed for
likelihood estimation and computations of proxies.}



\section*{Acknowledgments}

This research has been support with an NSERC Discovery Grant.


\renewcommand{\thesection}{\Alph{section}}
\setcounter{section}{0}

\section*{Appendices}
\section{Derivations: non-theorems}

\subsection{Bi-factor Gaussian model: Equivalence of two-stage factor scores defined in \eqref{eq-bifactor-fs-W0} and \eqref{eq-bifactor-fs-Wg} and regression factor scores}
\label{sec:equivalence}

\begin{proof}
Suppose there are $G$ groups, and let $\zb_{\Db}=(\zb_1^{T},\zb_{2}^{T},\ldots,\zb_{G}^{T})^{T}$, where $\zb_1,\ldots,\zb_{G}$ are the realization of observed variables $\Zb_1,\ldots,\Zb_{G}$. 
Let $\what_{0}=\e(W_0|\Zb_{\Db}=\zb_{\Db})$. For proxies of local latent factors, let $\what_{1}=\e(W_1|\Zb_{\Db}=\zb_{\Db})$. 
$\wtil_{0},\wtil_{1}(\wtil_{0})$ are defined in \eqref{eq-bifactor-fs-W0} and \eqref{eq-bifactor-fs-Wg}. 
Let $\wtil_{1}=\wtil_{1}(\wtil_{0})$ for notation simplicity. With loss of generality, it suffices to prove that $\wtil_{0}=\what_0$, $\wtil_{1}=\what_{1}$, as all indices of local latent factors could be permuted to be in the first group. 
Let $\Db=(d_1,d_2,\ldots,d_{G})$, $d=\sum_{j=1}^{G}d_g$, $d_{r}=\sum_{j=2}^{G}d_2$, $\zb_{D}=(\zb_1^{T},\zb_{r}^{T})^{T}$.
Then the loading matrix 
\begin{equation}\label{eqn:loading}
 \Ab=[\ab_0,\ab_1,\ldots,\ab_{p}]=\begin{bmatrix}
 \begin{array}{cc|cccc}
\bd_{01} & \bd_{1} & \bm{0} & \ldots & \bm{0}\\
\hline
\bd_{02} & \bm{0} & \bd_{2} & \ldots & \bm{0}\\
\vdots & \vdots & \vdots & \vdots\\
\bd_{0G} & \bm{0} & \ldots & \bm{0} & \bd_{G}\\
\end{array}
\end{bmatrix}=
\begin{bmatrix}
\bd_{01} & \bd_{1} & \bm{0}\\
\bd_{0r} & \bm{0} & {\bm{B}_{r}}\\
\end{bmatrix}.
\end{equation}

The partition of $\Ab$ leads to a $2\times 2$ block matrix, 
where $\bd_{0r}=(\bd_{02}^{T},\ldots,\bd_{0G}^{T})^{T}$, 
$\bm{B}_{r}=\diag(\bd_2,\ldots,\bd_{G})$. 
Also, partition $\Psib^{2}=\diag(\Psib^{2}_1,\ldots,\Psib^{2}_{G})=\diag(\Psib^{2}_1,\Psib^{2}_{r})$ correspondingly.\\
Let
$\sigb_{D}=\Cor(\Zb_{\Db})=
\begin{bmatrix}
{\sigb}_{11} & {\sigb}_{12}\\
{\sigb}_{21} & {\sigb}_{22}
\end{bmatrix}
$=$\begin{bmatrix}
\bd_{01}\bd_{01}^{T}+\bd_{1}\bd_{1}^{T}+\Psib_{1}^{2} & \bd_{01}\bd_{0r}^{T}\\
\bd_{0r}\bd_{01}^{T} & {\bd_{0r}\bd_{0r}^T+} \bm{B}_{r}\bm{B}_{r}^{T}+\Psib_{r}^{2}\\
\end{bmatrix}$, 
$\bm{M}=\sigb_{\Db}^{-1}=:
\begin{bmatrix}
\bm{M}_{11} & \bm{M}_{12}\\
\bm{M}_{21} & \bm{M}_{22}
\end{bmatrix}.$
Then the sizes of matrices $\bm{M}_{11}, \bm{M}_{12}, \bm{M}_{21}, \bm{M}_{22}$ are $d_1\times d_1$, $d_1\times d_r$, $d_r\times d_1$, $d_r\times d_r$ respectively, and the corresponding blocks in $\sigb_{\Db}$ have the same size. 
{Let $\sigb_{1}$ be the correlation matrix of $(\Zb_1^T,W_0)$.}
Then $\sigb_{1}=\begin{bmatrix}
\bd_{01}\bd_{01}^{T}+\bd_{1}\bd_{1}^{T}+\Psib_{1}^{2} & \bd_{01}\\
\bd_{01}^{T} & 1
\end{bmatrix}
$=:$
\begin{bmatrix}
\sigb_{11} & \sigb_{12}\\
\sigb_{21} & \sigb_{22}
\end{bmatrix}$, and $\bm{N}=\sigb_{1}^{-1}=
\begin{bmatrix}
\bm{N}_{11} & \bm{N}_{12}\\
\bm{N}_{21} & \bm{N}_{22}
\end{bmatrix}$. 
The sizes of the matrices $N_{11}, N_{12}, N_{21}$ are $d_1\times d_1$, $d_1\times 1$, $1 \times d_1$ respectively, and $N_{22}$ is a scalar. 
Let $\zb_{D}=(\zb_1^{T},\zb_{r}^{T})^{T}$, the regression factor scores defined in \eqref{eq-1factor-fs} are:
$\what_{0}=\ab_{0}^{T}\sigb_{D}^{-1}\zb_{D}=\ab_{0}^{T}\bm{M}\zb_{D}, \quad \what_{1}=\ab_{1}^{T}\sigb_{D}^{-1}\zb_{D}=\ab_{1}^{T}\bm{M}\zb_{D}$. Hence, by \eqref{eqn:loading}
\begin{align}
\label{eqn: w0}
\what_0 & =(\bd_{01}^{T},\bd_{0r}^{T})
\begin{bmatrix} \bm{M}_{11} & \bm{M}_{12}\\
\bm{M}_{21} & \bm{M}_{22}
\end{bmatrix}
\left(\begin{array}{c}   
 \zb_1\\ \zb_r \end{array} \right) =  
(\bd_{01}^{T}\bm{M}_{11}+\bd_{0r}^{T}\bm{M}_{21})\zb_1+(\bd_{01}^{T}\bm{M}_{12}+\bd_{0r}^{T}\bm{M}_{22})\zb_{r},\\
\label{eqn: w1}
\what_{1} & =(\bd_1^{T},\bm{0}_{r})
\begin{bmatrix}
\bm{M}_{11} & \bm{M}_{12}\\
\bm{M}_{21} & \bm{M}_{22}
\end{bmatrix}
\left(\begin{array}{c}   
\zb_1\\ \zb_r \end{array}
\right)=\bd_{1}^{T}\bm{M}_{11}\zb_1+\bd_{1}^{T}\bm{M}_{12}\zb_r.
\end{align}
The expressions of $\wtil_0$ and $\what_0$ equal $\ab_{0}^{T}\bm{M}\zb_{\Db}$, so they are the same. 
After some algebraic calculations {in \eqref{eq-bifactor-fs-Wg}}, $\wtil_1=\bd_{1}^{T}\bm{N}_{11}\zb_1+\bd_{1}^{T}\bm{N}_{12}\wtil_{0}$. 
Substituting  $\what_{0}=\wtil_0$ {from {\eqref{eqn: w0}} into 
$\wtil_1$ leads to}
\begin{equation}
\label{eqn: tem1}
  \wtil_{1}=[\bd_{1}^{T}\bm{N}_{11}+\bd_{1}^{T}\bm{N}_{12}(\bd_{01}^{T}\bm{M}_{11}
 +\bd_{0r}^{T}\bm{M}_{21})]\zb_1+\bd_{1}^{T}\bm{N}_{12}(\bd_{01}^{T}\bm{M}_{12}+\bd_{0r}^{T}\bm{M}_{22})\zb_{r}.
\end{equation}
To conclude, it suffices to show that $\wtil_1$ in \eqref{eqn: tem1} and $\what_1$ in \eqref{eqn: w1} are equivalent, or that
\\
(a) $\bm{N}_{11}+\bm{N}_{12}(\bd_{01}^{T}\bm{M}_{11}+\bd_{0r}^{T}\bm{M}_{21})=\bm{M}_{11}$, and
(b) $\bm{N}_{12}(\bd_{01}^{T}\bm{M}_{12}+\bd_{0r}^{T}\bm{M}_{22})=\bm{M}_{12}$.

Let $\Delta_{1}=(\bd_1\bd_1^{T}+\Psib_{1}^{2})$, 
$\Delta_1$ is positive definite and  $\Delta_1+\bd_{01}\bd_{01}^{T} =\sigb_{11}$.
Multiply $\Delta_1^{-1}$ on the left and $\sigb_{11}^{-1}$ on the right to get  (c) $\Delta_1^{-1}\bd_{01}\bd_{01}^{T}\sigb_{11}^{-1}-\Delta_{1}^{-1}=-\sigb_{11}^{-1}$.
From $\sigb_{D}\bm{M}=\bigi$, we have (d) 
$\bm{M}_{11}=\sigb_{11}^{-1}-\sigb_{11}^{-1}\bd_{01}\bd_{0r}^{T}\bm{M}_{21}$
and (e)
$\bm{M}_{12}=-\sigb_{11}^{-1}\bd_{01}\bd_{0r}^{T}\bm{M}_{22}$. 
From $\bm{N}\sigb_{1}=\bigi$, we have (f) $\bm{N}_{11}\sigb_{11}+\bm{N}_{12}\bm{b}_{01}^{T}=\bm{N}_{11}(\Delta_1+\bd_{01}\bd_{01}^{T})+\bm{N}_{12}\bd_{01}^{T}=\bigi$ and (g) $\bm{N}_{11}\bd_{01}+\bm{N}_{12}=0$. 
In (g), multiply both sides by $\bd_{01}^{T}$ to get (h) 
$\bm{N}_{11}\bd_{01}\bd_{01}^{T}+\bm{N}_{12}\bd_{01}^{T}=0$. 
Then (f) and (h) together imply (i) $\bm{N}_{11}\Delta_{1}=\bigi$. 
Hence,  from (g) and (i), $\bm{N}_{12}=-\bm{N}_{11}\bd_{01}=-\Delta_1^{-1}\bd_{01}$,
and from (f), $\bm{N}_{11}=\sigb_{11}^{-1}-\bm{N}_{12}\bd_{01}^{T}\sigb_{11}^{-1}=\sigb_{11}^{-1}+\Delta_1^{-1}\bd_{01}\bd_{01}^{T}\sigb_{11}^{-1}$. 
Substitute these expressions of $\bm{N}_{11}$ and $\bm{N}_{12}$ in the left-hand side of equation (a) to get: 
\begin{align*}
 \bm{N}_{11}&+\bm{N}_{12}\bd_{01}^{T}\bm{M}_{11}+\bm{N}_{12}\bd_{0r}^{T}\bm{M}_{21}
 =\sigb_{11}^{-1}+ \bigl\{\Delta_1^{-1}\bd_{01}\bd_{01}^{T}\sigb_{11}^{-1}
 -\Delta_1^{-1}\bd_{01}\bd_{01}^{T}\bm{M}_{11}
 -\Delta_1^{-1}\bd_{01}\bd_{0r}^{T}\bm{M}_{21}\bigr\}.
\end{align*}
For the right-hand side of the above, substitute $\bm{M}_{11}$ from (d) and
then $-\sigb_{11}^{-1}$ in (c), 
so that the sum of the last three terms in braces becomes 
\begin{align*}
&\Delta_1^{-1}\bd_{01}\bd_{01}^{T}\sigb_{11}^{-1}
 -\Delta_1^{-1}\bd_{01}\bd_{01}^{T}(\sigb_{11}^{-1}
 -\sigb_{11}^{-1}\bd_{01}\bd_{0r}^{T}\bm{M}_{21})
 -\Delta_1^{-1}\bd_{01}\bd_{0r}^{T}\bm{M}_{21} \\
&=\Delta_1^{-1}\bd_{01}\bd_{01}^{T}\sigb_{11}^{-1}\bd_{01}\bd_{0r}^{T}\bm{M}_{21}
 -\Delta_1^{-1}\bd_{01}\bd_{0r}^{T}\bm{M}_{21} \\
 &=(\Delta_1^{-1}\bd_{01}\bd_{01}^{T}\sigb_{11}^{-1}
 -\Delta_{1}^{-1})\bd_{01}\bd_{0r}^{T}\bm{M}_{21}
 =-\sigb_{11}^{-1}\bd_{01}\bd_{0r}^{T}\bm{M}_{21}.
\end{align*}
Thus, (a) is verified as $\bm{N}_{11}+\bm{N}_{12}\bd_{01}^{T}\bm{M}_{11}+\bm{N}_{12}\bd_{0r}^{T}\bm{M}_{21}
 =\sigb_{11}^{-1}-\sigb_{11}^{-1}\bd_{01}\bd_{0r}^{T}\bm{M}_{21}=\bm{M}_{11}$
via (d).
Next, substitute $\bm{N}_{12}$ in (g) and $\bm{M}_{12}$ in (e) in the left-hand side of (b), so that (b) is verified as 
\begin{align*}
\bm{N}_{12}(\bd_{01}^{T}\bm{M}_{12}+\bd_{0r}^{T}\bm{M}_{22})&
 =\Delta_1^{-1}\bd_{01}\bd_{01}^{T}\sigb_{11}^{-1}\bd_{01}\bd_{0r}^{T}\bm{M}_{22}
 -\Delta_{1}^{-1}\bd_{01}\bd_{0r}^{T}\bm{M}_{22}\\
&=(\Delta_1^{-1}\bd_{01}\bd_{01}^{T}\sigb_{11}^{-1}
 -\Delta_{1}^{-1})\bd_{01}\bd_{0r}^{T}\bm{M}_{22}
 =-\sigb_{11}^{-1}\bd_{01}\bd_{0r}^{T}\bm{M}_{22}=\bm{M}_{12},
\end{align*}
via (c) and then (e).
\end{proof}

\subsection{Proof for \eqref{decomp_0} in Table 2}
\label{sec: decomp}

In the bi-factor model \eqref{eq-bifactor-gauss},  
recall the notations $\bd_g$, $\bd_{0g}$,$\sigb_{g}$ defined in Section  
\ref{sec-factorscore},
and $q_{g}$, $\Tilde{q}_{g}$ defined in the caption of Table \ref{tab:cond_var}.   
Let $\bm{r}_g=(\bd_g^{T},0)\sigb_g^{-1}$ be a vector of length $d_g+1$. Let $r$ be the last entry of vector $\bm{r}_g$. 
Define $\sigb_{gg}=\bm{B}_g\bm{B}_g^{T}+\Psib_g^{2}$. 
Let the last column of $\sigb_{g}^{-1}$ be $[\bm{s}_{12}^{T},s_{22}]^{T}$. From $\sigb_{g}\sigb_{g}^{-1}=\bigi$, two equations are obtained. (a) $\sigb_{gg}\bm{s}_{12}+\bd_{0g}\bm{s}_{22}=0$ and (b) $\bd_{0g}^{T}\bm{s}_{12}+\bm{s}_{22}=1$. Multiply both sides of (b) by $\bd_{0g}$ to get (c) $\bd_{0g}\bd_{0g}^{T}\bm{s}_{12}+\bd_{0g}\bm{s}_{22}=\bd_{0g}$. Then (a) and (c) together implies $(\bd_{0g}\bd_{0g}^{T}-\sigb_{gg})\bm{s}_{12}=\bd_{0g}$. Hence, $\bm{s}_{12}=(\bd_{0g}\bd_{0g}^{T}-\sigb_{gg})^{-1}\bd_{0g}$. Since
$\bm{r}_{g}=(\bd_g^{T},0)\sigb_g^{-1}=(\bd_g^{T},0)\begin{bmatrix}
*,\bm{s}_{12}\\
*,s_{22}
\end{bmatrix}=
[*,\bd_{g}^{T}\bm{s}_{12}]$, {the last entry of $\bm{r}_{g}$ is $r=\bd_{g}^{T}\bm{s}_{12}$}.

{In the definition of factor scores in \eqref{eq-bifactor-fs-Wg}, 
$\e(W_g|\Zb_{\Db},W_0)=\bm{r}_{g}(\Zb_{g}^{T},W_0)^{T}=h(\Zb_{g})+rW_0$,
where $h(\Zb_{g})$ is a linear function of $\Zb_{g}$.
Then $\Var[\e(W_g|\Zb_{\Db},W_0)|\Zb_D]=r^2\Var(W_0|\Zb_D)$.}
From the conditional variance decomposition formula,
\begin{align}
\Var(W_g|\Zb_{\Db}) & =\e[\Var(W_g|\Zb_{\Db},{W}_{0})|\Zb_{\Db}]
 +\Var[\e(W_g|\Zb_{\Db},W_0)|\Zb_{\Db}]
\notag\\
& =\underbrace{[1-\bd_{g}^{T}(\bd_g\bd_{g}^{T}+
\Psib_g^{2})^{-1}\bd_{g}]}_{\text{term1}}
 +\underbrace{r^{2}\Var(W_0|\Zb_{\Db})}_{\text{term2}},
\label{decomp1}
\end{align}
where $r=\bd_{g}^{T}\bm{s}_{12}=-\bd_{g}^{T}(\sigb_{gg}
-\bd_{0g}\bd_{0g}^{T})^{-1}\bd_{0g}=-\bd_g^{T}(\bd_g\bd_g^{T}+\Psib_g^{2})^{-1}\bd_{0g}.$ 

{Since $\Zb_g$ independent of $Z$'s in other
groups given $W_0$, $\Var(W_g|\Zb_{\Db},W_0)=\Var(W_g|\Zb_{g};W_0)$}.
The term1 in \eqref{decomp1} follows because the joint distribution of $(\Zb_{g}^{T},W_g)^{T}$ given $W_0$ is multivariate normal with zero mean and covariance matrix
$\begin{bmatrix}
\bd_g\bd_g^{T}+\Psib_g^{2} & \bd_g\\
\bd_g^{T} & 1\\
\end{bmatrix} $.
From assuming $\psi_{jg}>0$ for all $j,g$, applying \eqref{eq-matrix-id} {with $\Ab_D=\bd_g$ and  $\Psib_D=\Psib_g$}, 
term1 simplifies into $(1+\bd_g^{T}\Psib_g^{-2}\bd_g)^{-1}=(1+q_g)^{-1}$.  
As for term2 in \eqref{decomp1}, 
in the expression of $r$, applying equation $\eqref{eq-matrix-id}$ as above,  
$r=-\Tilde{q}_g (1+q_g)^{-1}$. 
Combine the expression of two terms, the decomposition \eqref{decomp_0} is
obtained.

\section{Main Proofs in Sections \ref{sec-consistency1} and \ref{sec-residdep}}
\label{sec-1factor-consistency}

\subsection{Proof of Theorem \ref{1_factor_gaussian}}
\label{sec: theorem_1_fact}

\begin{proof}
In 1-factor model ($p=1$), the loading matrix $\Ab_{D}$ is $D\times 1$, so we use notation $\Ab_{D}$ instead. 
Due to assumption on $\alpha_j$ uniformly bounded away from $\pm$1, $\Psib_D^{-1}$ is well-defined for all $D$. 
Thus the regression factor scores can be expressed in two equivalent forms. 
In the expression \eqref{eq-1factor-fs}, let $q_{D}=\Ab_{D}^{T}\Psib_{D}^{-2}\Ab_{D}$ (a positive real number). 
Since
$D^{-1}\sum_{j=1}^{D}|\alpha_{j}|\leq (D^{-1}\sum_{j=1}^{D}\alpha_{j}^{2})^{1/2}\leq (D^{-1}\sum_{j=1}^{D}|\alpha_{j}|)^{1/2}$ and  
$\lim_{D\to\infty}D^{-1}\sum_{j=1}^{D}|\alpha_{j}|\to\text{const}\neq 0$, then
$\bar{q}_{D}:=D^{-1}q_{D}=D^{-1}\sum_{j=1}^{D}\alpha_{j}^{2}/\psi_{j}^{2}\to q >0$ {(with the limit existing assuming sampling from a super-population).}
Since
\begin{align}
\wtil_{D}&=(1+q_{D})^{-1}\ad^{T}\Psib_{D}^{-2}(\ad w^{0}+\Psib_{D}\bm{e}_{D})
\notag\\
&=(1+q_{D})^{-1}\ad^{T}\Psib_{D}^{-1}\bm{e}_{D}+(1+q_{D}^{-1})^{-1}w^{0},
 \label{eq-1factor-regscore}
\end{align}
then $\wtil_{D}-w^{0}={(1+q_{D})^{-1}\ad^{T}\Psib_{D}^{-1}\bm{e}_{D}} +O(D^{-1})$.
Next, $D^{-1/2}\Ab_{D}^{T}\Psib_{D}^{-1}\bm{e}_{D}$ 
is a realization of $D^{-1/2}\Ab_{D}^{T}\Psib_{D}^{-1}\epsb_{D}=D^{-1/2}\sum_{j=1}^{D}\alpha_{j}\epsilon_{j}/\psi_{j}$
which converges to $N(0,q)$ in distribution by the Continuity theorem, 
so it can be consider as $O_{p}(1)$.
Hence, 
\begin{equation*}
(\wtil_{D}-w^{0})=D^{-1/2}(D^{-1}+\bar{q}_{D})^{-1}(D^{-1/2}\Ab_{D}^{T}
\Psib_{D}^{-1}\bm{e}_{D})+O(D^{-1})
\end{equation*}
is asymptotically $O_{p}(D^{-1/2})$.
\end{proof}

\subsection{Extension to weak residual dependence}
\label{sec-proof-residdep}

An outline of the proof of consistency based on  Assumption \ref{errors}
is as follows.

Let $\bm{e}_{D}=(e_1,\ldots,e_{D})$ be one realization of $\epsb_{D}$. Let  $q_{D}=\ad^{T}\Psib_{D}^{-2}\ad>0$, and suppose $D^{-1}q_{D}\to q>0$.
Then, as in \eqref{eq-1factor-regscore}
$\wtil_{D}-w^{0}=(1+q_{D})^{-1}\ad^{T}\Psib_{D}^{-1}\bm{e}_{D}+(1+q_{D}^{-1})^{-1}w^{0}-w^{0}$. 
Note that
$Y=D^{-1/2}\Ab_{D}^{T}\Psib_{D}^{-1}\bm{e}_{D}$ 
is a realization of $D^{-1/2}\Ab_{D}^{T}\Psib_{D}^{-1}\epsb_{D}=D^{-1/2}\sum_{j=1}^{D}\alpha_{j}\epsilon_{j}/\psi_{j}$ {with variance
$D^{-1}\sum_{j=1}^D\sum_{k=1}^D\alpha_j\alpha_k\omega_{jk}/[\psi_j\psi_k]$.}
{By Assumption $\ref{errors}$ and with loadings that are bounded away from $\pm1$, this variance is $O(1)$ so that $Y$
can be considered as $O_{p}(1)$.} 
Then $\wtil_{D}-w^{0}=D^{-1/2}(D^{-1}+\bar{q}_{D})^{-1}(D^{-1/2}\Ab_{D}^{T}\Psib_{D}^{-1}\epsb_{D})+O(D^{-1})$ is asymptotically $O_{p}(D^{-1/2})$.

\subsection{Proof of Theorem: \ref{bi_fact_consistent}}
\label{sec: theorem_bi_fact}

\begin{proof}
Using the technique in the proof of Theorem \ref{1_factor_gaussian}, let $\Qb_{\Db}=\Ab_{\Db}^{T}\Psib_{\Db}^{-2}\Ab_{\Db}$, 

Then 
\begin{align*}
    \wtilb_{\Db}-\wb^{0}={(\bigi_{p}+\Qb_{\Db})^{-1}\Ab_{\Db}^{T}\Psib_{\Db}^{-1}\bm{e}_{\Db}}+{(\bigi_{p}+\Qb_{\Db}^{-1})^{-1}\wb^{0}-\wb^{0}}
\end{align*}
Since $\Ab_{\Db}$ is of full rank and the entries of $\Psib_{g}$ for $g=1,2,\ldots,G$ are uniformly bounded away from 0. 
Then $\bar{\Qb}_{\Db}=D^{-1}\Ab_{\Db}^{T}\Psib_{\Db}^{-2}\Ab_{\Db}$ is positive definite for any fixed $D$,
and $\Qb:= \lim_{\Db\to\infty}\bar{\Qb}_{\Db}$ must be a semi positive definite matrix, {with the limit existing assuming sampling from a super-population}.
Since $D^{-1}\|\ab_j\|_1\not\to 0$, then $\Qb$ is a positive definite matrix.
{Since $\epsb_{\Db}\sim \bm{N}(\bm{0},\bm{I}_{\Db})$,}
then 
$D^{-1/2}\Ab_{\Db}^{T}\Psib_{\Db}^{-1}\epsb_{\Db}\to N(0,\Qb)$. Since $D^{-1/2}\Ab_{\Db}^{T}\Psib_{\Db}^{-1}\bm{e}_{\Db}$ is one realization, it can considered as $O_{p}(1)$. 
Hence,
\begin{align*}
\wtilb_{\Db}-\wb^{0}=D^{-1/2}(D^{-1}\bigi_p+\bar{\Qb}_{\Db})^{-1}D^{-1/2}\Ab_{\Db}^{T}\Psib_{\Db}^{-1}\epsb_{\Db}+O(D^{-1})
\end{align*}
is asymptotically $O_{p}(D^{-1/2})$ by noticing that
$(D^{-1}\bigi_p+\bar{\Qb}_{\Db})^{-1}\to \Qb^{-1}$.
\end{proof}

\subsection{Proof of consistency for proxies: Theorem \ref{consistent_proxy_1fact}, Theorem \ref{consistency_proxy_bifct}}
\label{sec-proof-proxies}

For the conditional expectations for 1-factor and bi-factor copulas, the $v$'s should be treated as parameters, and the $u$'s are the realization of independent random variables when the latent variables are fixed. 
The proof techniques of Theorem \ref{consistent_proxy_1fact}
and Theorem \ref{consistency_proxy_bifct} are similar. 
Both rely on the Laplace approximation for integrals (see \citep{breitung1994asymptotic}), and the asymptotic properties of maximum likelihood (ML) estimator for parameters.
In our setting, the results in \cite{bradley1962asymptotic} are used for the asymptotics of a log-likelihood for a sample $X_j\sim f_{X_j}$ from independent but not identically distributed observations with common parameters over the $\{f_{X_j}\}$. 

The proof of Theorem \ref{consistent_proxy_1fact} is given below.

\begin{proof}
In 1-factor copula model \eqref{eq-1factor-copula-pdf}, there is a realized value $v^0$ for the latent variable.
Then $(U_1,\ldots,U_D,\ldots)$ is an infinite sequence of independent
random variables with $U_j\sim c_{jV}(\cdot,v^0)$.
If the value of $v^0$ is to be estimated based on the realized
$D$-vector $(u_1,\ldots,u_D)$, then the averaged
negative log-likelihood in $v$ is
  $$g_D(v)= -D^{-1} \sum_{j=1}^D \log c_{jV}(u_j,v).$$
The maximum likelihood estimate $v^*_D$ satisfies 
{$v^*_D=v^0+O_{p}(D^{-1/2})$} 
from results in \cite{bradley1962asymptotic}.
Now apply the Laplace approximation. The numerator and denominator denoted as $I_{1D}$ and $I_{2D}$ in the expression of $\vtil_{D}$ in \eqref{eq-1factor-fs} can be approximated respectively by 
\begin{align*}
    I_{1D}&=\int_{0}^{1}v \exp\{-D\times g_{D}(v)\}\dd v
    =v_{D}^{*}\exp\{-D\times g(v_{D}^{*})\}\sqrt{\frac{2\pi}{D |g^{''}(v_{D}^{*})|}} \,\big\{1+O(D^{-1})\big\},\\
    I_{2D}&=\int_{0}^{1}\exp\{-D\times g_{D}(v)\}\dd v
    =\exp\{-D\times g_{D}(v_{D}^{*})\}\sqrt{\frac{2\pi}{D |g^{''}(v_{D}^{*})|}}
  \,\{1+O(D^{-1})\}.
\end{align*}
Hence $\vtil_{D}-v_{D}^{*}=I_{1D}/I_{2D}-v_{D}^{*}=O(D^{-1})$
and {$\vtil_{D}-v^0=O_p(D^{-1/2})$.}
\end{proof}

The proof of Theorem \textbf{\ref{consistency_proxy_bifct}} for bi-factor copula model is given next.

\begin{proof}
There are realized value $v_0^0,v_1^0,\ldots,v_G^0$ for the latent variables.
Then $(U_{1g},\ldots,U_{d_gg},\ldots)$ is an infinite sequence of
dependent random variables for each $g=1,\ldots,G$, and the $G$ sequences
are mutually independent given latent variables.
For the bi-factor copula model, from Algorithm 24 in Joe (2014), the cdf of $U_{jg}$ is 
$C_{\Ugj|\Vg;\Vo}\bigl(C_{\Ugj|\Vo}(\cdot|v_0^0),v_g^0 \bigr)$
and its density is
 $$c_{\Ugj \Vo}(\cdot,v_0^0)\cdot
   c_{\Ugj,\Vg;\Vo}\bigl(C_{\Ugj|\Vo}(\cdot|v_0^0),v_g^0 \bigr).$$
If the values of $v^0,v_1^0,\ldots,v_G^0$ are to be estimated based on 
the realized
$d_g$-vector $\ub_{g,d_g}=(u_{1g},\ldots,u_{d_gg})$ for $g\in\{1,\ldots,G\}$,
then the integrated log-likelihood in $v_0$  is
  $$L_0(v_0) = \log c_{\Ub_{\Db},\Vo}(\ub_{\Db},\vo),$$
where $\log c_{\Ub_{\Db},\Vo}(\ub_{\Db},\vo)$ is defined in Table 1.
Take the partial derivative with respect to $\vo$ leads the first
inference function $\Psi_{0,D}(v_0;\ub_{\Db})$.
\begin{align*}
{\Psi}_{0,D}(v_{0};\ub_{\Db})&:
 =\partial \log c_{\Ub_{\Db},\Vo}(\ub_{\Db},\vo)/\partial\vo 
 ={\Psi}_{01,D}(v_0;\ub_{\Db})+{\Psi}_{02,D}(v_0;\ub_{\Db})\\
&=\sum_{g=1}^{G}\sum_{j=1}^{d_g}\partial \log c_{\Ugj,\Vo}(\ugj,\vo)
 /\partial \vo+\sum_{g=1}^{G}\partial \log f_{g}(\ub_g;\vo)/\partial \vo.
\end{align*}
Let $v^*_{0,\Db}$ be the maximum likelihood estimate and assume it is the
unique solution of $\bar{\Psi}_{0,D}=D^{-1}\Psi_{0,D}$. {Note that
assuming regularity assumptions include the exchange of
integration and the partial differentiation,}
\begin{align*}
{\partial \log f_{g}(\ub_{g};\vo)\over \partial \vo}
&=\frac{1}{f_{g}(\ub_g;v_0)}\intt  \Bigl(\partial\exp\Bigl\{\sum_{j=1}^{d_g}\log c_{\Ugj,\Vg;\Vo}(C_{\Ugj|\Vo}
 ( u_{jg}|v_0),v_g)\Bigr\}\Bigm/\partial v_0\Bigr) \dd v_g\\
&=\sum_{j=1}^{d_g}\intt \frac{c_{\Ub_{g},V_g;V_0}(\ub_g,v_g;v_0)}
  {f_{g}(\ub_g;v_0)}\frac{\partial \log c_{\Ugj V_g;V_0}(C_{\Ugj|\Vo}
  (u_{jg}|v_0),v_g)}{\partial v_0}\dd v_g.
\end{align*}

The derivatives of $\bar{\Psi}_{0,D}$ can be written as $\partial \bar{\Psi}_{0,D}(v_0)/
\partial v_0:=\partial \bar{\Psi}_{01,D}(v_0)/\partial v_0+\partial \bar{\Psi}_{02,D}(v_0)/\partial v_0$. 
From the laws of large numbers in page 174 of \cite{cramer1947problems}, under regularity conditions of the log-likelihood {and assuming sampling from a super-population},
\begin{align*}
&\lim_{\Db\to\infty}\bar{\Psi}_{01,D}({{v}_0},\ub_{D})
=\lim_{D\to\infty}D^{-1}\sum_{g=1}^{G}\sum_{j=1}^{d_g}\partial\log c_{\Ugj,\Vo}(\ugj,\vo)/\partial \vo,  \\
&\lim_{\Db\to\infty}\bar{\Psi}_{02,D}({{v}_0},\ub_{\Db})
=\lim_{D\to\infty}D^{-1}\sum_{g=1}^{G}\sum_{j=1}^{d_g}
\intt \frac{c_{\Ub_{g},V_g;V_0}(\ub_g,v_g;v_0)}{f_{g}(\ub_g;v_0)}
\frac{\partial \log c_{\Ugj V_g;V_0}(C_{\Ugj|\Vo}( u_{jg}|v_0),v_g)}{\partial v_0}\dd v_g
\end{align*}
exist. {With enough dependence on the latent variable,
the derivative of $\bar{\Psi}_{0,D}(v_0,\ub_{\Db})$ is bounded away from 0.}
Then as $d_g\to\infty$ for all $g$, $v^*_{0,\Db}=v^0_0+o_p(1)$.

Furthermore, the profile log-likelihood in $v_g$ given $v_0$ is
(from Table 1):
  $$L_g(v_g;\ub_{g,d_g},v_0) = \log c_{\Ub_{g},\Vo,\Vg}(\ub_{g,d_g},\vo,\vg)
 = \sum_{j=1}^{d_g} \log 
  \Bigl\{ c_{\Ugj \Vo}(u_{jg},v_0) + \log
   c_{\Ugj|\Vg;\Vo}\bigl(C_{\Ugj|\Vo}(u_{jg}|v_0),v_g \bigr)
  \Bigr\}.$$
The partial derivative of $L_g$ with respect to $v_g$ 
leads to the inference function $\Psi_{g,D}(v_g;\ub_{g,d_g},v_0)$.
For $v_0$ in a neighborhood of {$v^*_{0,\Db}$},
let $v_{g,d_g}^*(v_0)$ be maximum profile likelihood estimate and assume it is 
the unique solution of 
{$\bar{\Psi}_{g,D}(v_g;\ub_{g,d_g},v_{0}):=d_{g}^{-1}{\Psi}_{g,D}(v_g;\ub_{g,d_g},v_{0})$}.
From the weak laws of large numbers, {and the super-population
assumption}, 
\begin{equation*}
\lim_{d_{g}\to\infty}\bar{\Psi}_{g,D}(v_g;\ub_{g,d_g},v_{0,\Db}^{*})
=\lim_{d_g\to\infty}d_{g}^{-1}\sum_{j=1}^{d_g}
{ {\partial \log c_{\Ugj|\Vg;\Vo}\bigl(C_{\Ugj|\Vo}(\ugj|v_{0,\Db}^{*})
\bigg|\vg\bigr) \over \partial v_g} }
\end{equation*}
exists.
As $d_g\to\infty$ for all $g$, $v^*_{g,d_g}(v^*_{0,\Db})=v^0_g+o_p(1)$.

For the proxy defined in \eqref{eq-bifactor-gproxy},
\begin{equation}
\label{laplace1}
\vtil_{0\Db}(\ub_{\Db})=\frac{\intt v_0 \exp\{-D\cdot \bar{L}_{0\Db}(v_0;\ub_{\Db})\}\dd v_0}{\intt \exp\{-D\cdot \bar{L}_{0\Db}(v_0;\ub_{\Db})\}\dd v_0}.
\end{equation}
Since $\bar{L}_{0\Db}$ attains the global minimum at $v_{0,\Db}^{*}$, then from the Laplace approximation, when $\Db\to\infty$, the numerator and denominator in \eqref{laplace1} can be approximated by
\begin{align*}
   & v_{0,\Db}^* \exp\{-D\times\bar{L}_{0\Db}(v_{0,\Db}^*)\} \sqrt{\frac{2\pi}
  { {D\bigg|[\partial\bar{\Psi}_{0,D}/\partial v_0]|_{v_{0,\Db}^*}}\bigg| }}
  +O(D^{-1}),\\
   &\exp\{-D\times\bar{L}_{0\Db}(v_{0,\Db}^*)\} \sqrt{\frac{2\pi}
  {D\bigg|[\partial {\bar\Psi}_{0,D}/\partial v_0]|_{v_{0,\Db}^*}\bigg|} } +O(D^{-1})
\end{align*}
respectively.
Then, $\vtil_{0\Db}(\ub_{\Db})=v_{0,\Db}^{*}+O(D^{-1})$.
Similarly {for \eqref{eq-bifactor-lproxy}}, 
from the Laplace approximation, 
${\vtil_{g\Db}(\ub_{g,d_g};\vtil_{0\Db})=}
v_{g,d_g}^{*}(v_{0,\Db}^{*})+O(D^{-1})$ for all $g$. 
Thus, the proxies $\vtil_{0\Db}$ and $\vtil_{g\Db}$ for $g=1,\ldots,G$ are consistent.
\end{proof}

\subsection{Proof of Theorem \ref{lip_factormodel} (Lipschitz continuity of factor scores in Gaussian factor model)}

\begin{proof}
The difference between $\wtilb_{\Db}(\hatAb_{\Db})$ and $\wtilb_{\Db}(\Ab_{\Db})$ can be written as 
\begin{align}
\label{dif1}
&(\bigi_{p}+\widehat{\Qb}_\Db)^{-1}\hatAb_{\Db}^{T}\widehat{\Psib}_{\Db}^{-2}
\zb_{\Db}-(\bigi_{p}+{\Qb}_\Db)^{-1}\Ab_{\Db}^{T}{\Psib}_{\Db}^{-2}\zb_{\Db}\\
&=\underbrace{(\bigi_{p}+\widehat{\Qb}_\Db)^{-1}\hatAb_{\Db}^{T}
\widehat{\Psib}_{\Db}^{-2}\zb_{\Db}-(\bigi_{p}+{\Qb}_\Db)^{-1}
\hatAb_{\Db}^{T}\widehat{\Psib}_{\Db}^{-2}\zb_{\Db}}_{\text{term1}}\notag\\
&+\underbrace{(\bigi_{p}+{\Qb}_\Db)^{-1}\hatAb_{\Db}^{T}
\widehat{\Psib}_{\Db}^{-2}\zb_{\Db}-(\bigi_{p}
+{\Qb}_\Db)^{-1}\Ab_{\Db}^{T}{\Psib}_{\Db}^{-2}\zb_{\Db}}_{\text{term2}}\notag.
\end{align}
Let $\bar{\Qb}_{\Db}=D^{-1}\Qb_{\Db}$,
$\bar{\hatqb}_{\Db}=D^{-1}{\hatqb}_{\Db}$ and recall we assume 
$\bar{\Qb}_{\Db}\to\Qb$, $\bar{\hatqb}_{\Db}\to
\hatqb$ as $\Db\to\infty$ where $\Qb, \hatqb$ are positive definite matrix. 
Note $(\bigi_p+\Qb_{\Db})^{-1}=O(D^{-1})$, $(\bigi_p+\bar{\hatqb}_{\Db})^{-1}=O(D^{-1})$. 
Let $\hb_{\Db}=D\cdot (\bigi_p+{\Qb}_{\Db})^{-1}$, $\hb_{\Db}=O(1)$.
Since $\bar{\Qb}_{\Db}$ and $\bar{\hatqb}_{\Db}$ are both positive definite and well-conditioned,
then there is bound on the condition numbers of
$\bar{\Qb}_{\Db}$ and $\bar{\hatqb}_{\Db}$ for all large $D$, and 
  $$||\bar{\hatqb}_D^{-1}-\bar{\Qb}_D^{-1}||
  = O(||\bar{\hatqb}_D-\bar{\Qb}_D||) .$$
Then, term 1 in \eqref{dif1} {has the order of} $(\bar{\hatqb}_{\Db}^{-1}-\bar{\Qb}_{\Db}^{-1})\cdot D^{-1}\cdot \hatAb_{\Db}^{T}\widehat{\Psib}_{\Db}^{-2}\zb_{\Db}$.

For simplicity, we suppress the subscript of $\bar{\Qb}_{\Db}$, $\bar{\hatqb}_{\Db}$, $\Ab_{\Db}$ and $\Psib_{\Db}$ in the below derivation, 
\begin{align}
\label{dif_q}
  \bar{\hatqb}-\bar{\Qb}
  &= D^{-1}({\hatAb}^{T}\widehat{\Psib}^{-2}\hatAb-\Ab^{T}\Psib^{-2}\Ab)\notag\\
  &=D^{-1}\left((\hatAb^{T}-\Ab^{T}+\Ab^{T})\widehat{\Psib}^{-2}(\hatAb-\Ab+\Ab)-\Ab^{T}\Psib^{-2}\Ab\right)\notag\\
  &=D^{-1}\left(\underbrace{\Ab^{T}\widehat{\Psib}^{-2}(\hatAb-\Ab)}_{\text{term1}}
  +(\underbrace{\hatAb^{T}-\Ab^{T})\widehat{\Psib}^{-2}\Ab}_{\text{term2}}+
  \underbrace{(\hatAb^{T}-\Ab^{T})\widehat{\Psib}^{-2}(\hatAb-\Ab)}_{\text{term3}}+\underbrace{\Ab^{T}(\widehat{\Psib}^{-2}-\Psib^{-2})\Ab}_{\text{term4}}\right).
\end{align}
Since term3 is 
negligible compared to other terms, we only look at the order of the other three terms in the right-hand side of \eqref{dif_q}.
For term1, term2, term4, multiplied by $D^{-1}$, the Cauchy-Schwartz inequality leads to:
\begin{align*}
&\|D^{-1}\Ab^{T}\widehat{\Psib}^{-2}(\hatAb-\Ab)\|
  =\biggl\|D^{-1}\sum_{j=1}^{D}
\frac{\ab_j(\hat{\ab}_j-\ab_j)^{T}}{\widehat{\psi}_j^2}\biggr\|
  \leq 
\biggl(D^{-1}\sum_{j=1}^{D}\frac{\|\ab_j\|^2}{\widehat{\psi}_j^2}\biggr)^{1/2}
  \cdot\biggl(D^{-1}\sum_{j=1}^{D}{\frac{\|\hat{\ab}_j-\ab_j\|^2}{\widehat{\psi}_j^2}}\biggr)^{1/2},\\
&\|D^{-1}\Ab^{T}(\widehat{\Psib}^{-2}-\Psib^{-2})\Ab\|
  =\biggl\|D^{-1}\sum_{j=1}^{D}\ab_j\ab_j^{T}(\widehat{\psi}_{j}^{-2}
  -\psi_j^{-2}))\biggr\|
   \leq \biggl(D^{-1}\sum_{j=1}^{D}\|\ab_j\|^{4}\biggr)^{1/2}
  \cdot \biggl(D^{-1}\sum_{j=1}^{D}(\widehat{\psi}_{j}^{-2}-\psi_j^{-2})^2\biggr)^{1/2}.
\end{align*}
For all $j$, $\|\ab_j\|<1$, $\psi_j$ is bounded from zero, term1 and term2 in \eqref{dif_q} multiplied by $D^{-1}$ are all 
{$O\left(\sqrt{D^{-1}\sum_{j=1}^{D}{\|\hat{\ab}_j-\ab_j\|^2/\widehat{\psi}_j^2}}\right)$}. 
The term4 in \eqref{dif_q} multiplied by $D^{-1}$ is $O\left(\sqrt{D^{-1}\sum_{j=1}^{D}(\widehat{\psi}_{j}^{-2}-\psi_j^{-2})^2}\right)$. 
Also note 
\begin{align}
\label{cauchy}
   \| D^{-1}\hatAb_{\Db}^{T}\widehat{\Psib}_{\Db}^{-2}\zb_{\Db}\|
  =\biggl\|D^{-1}\sum_{j=1}^{D}\frac{\hat{\ab}_jz_j}{\widehat{\psi}_j^2}\biggr\|\leq 
  \sqrt{D^{-1}\sum_{j=1}^{D}\frac{\|
  \hat{\ab}_j\|^2}{\widehat{\psi}_j^4}}\sqrt{D^{-1}\sum_{j=1}^{D}p\cdot z_j^2} .
\end{align}
Since the bound of \eqref{cauchy} is $O_p(1)$, then term 1 in \eqref{dif1} is order 
${O_p}\left(\sqrt{D^{-1}\sum_{j=1}^{D}{\|\hat{\ab}_j-\ab_j\|^2/\widehat{\psi}_j^2}}\right)
+{O_p}\left(\sqrt{D^{-1}\sum_{j=1}^{D}(\widehat{\psi}_{j}^{-2}-\psi_j^{-2})^2}\right)$.

For term 2 in \eqref{dif1}, 
\begin{align}
\label{dif2}
 &\hb_{\Db}\cdot D^{-1}\cdot\left((\Ab_{\Db}^{T}+\hatAb_{\Db}^{T}
 -\Ab_{\Db}^{T})(\Psib_{\Db}^{-2}
  +\widehat{\Psib}_{\Db}^{-2}-\Psib_{\Db}^{-2})\zb_{\Db}-\Ab_{\Db}^{T}{\Psib}_{\Db}^{-2}\zb_{\Db}\right)\\
 &=\hb_{\Db}\cdot D^{-1}\cdot\left(\underbrace{\Ab^{T}_{\Db}(\widehat{\Psib}_{\Db}^{-2}  
 -\Psib_{\Db}^{-2})\zb_{\Db}}_{\text{term1}}+\underbrace{(\hatAb_{\Db}^{T}-\Ab_{\Db}^{T})\Psib^{-2}_{\Db}\zb_{\Db}}_{\text{term2}}
  +\underbrace{(\hatAb_{\Db}^{T}-\Ab_{\Db}^{T})(\widehat{\Psib}_{\Db}^{-2}
  -\Psib_{\Db}^{-2})\zb_{\Db}}_{\text{term3}}\right)\notag.
\end{align}
Since term3 in \eqref{dif2} is negligible in comparison, we only look at the first two terms.
From the Cauchy-Schwartz inequality, for term1 and term2 in \eqref{dif2} multiplied by $D^{-1}$, 
\begin{align*}
   D^{-1}\cdot\text{term1}= {O_p}
\left(\sqrt{D^{-1}\sum_{j=1}^{D}(\hat{\psi}_j^{-2}
  -\psi_{j}^{-2})^{2}}\right),\quad   D^{-1}\cdot\text{term2}
  {=} {O_p}\left(\sqrt{D^{-1}\sum_{j=1}^{D}\|\hat{\ab}_j-\ab_j\|^{2}}\right).
\end{align*}
Recall $\hb_{\Db}=O(1)$, then the term2 in \eqref{dif1} has the same order as term1 in \eqref{dif1}. Also, due to $\psi_j$ {being bounded away from 0} and 
$\psi_j^2=1-\|\ab_j\|^2$, then  $\|\wtil_{\Db}(\hatAb_{\Db})-\wtil_{\Db}(\Ab_{\Db})\|={O_p}\left(\sqrt{D^{-1}\sum_{j=1}^{D}{\|\hat{\ab}_j-\ab_j\|^2}}\right)$.
\end{proof}

\subsection{Proof of Lemma \ref{lipschitz}}

\begin{proof}
Let $\gamma: [0,1]\to \bar{\boldsymbol{B}}(\btheta_{D},\rho)$ be the path $\gamma(t)=t\hat{\btheta}_{D}+(1-t){\btheta}_{D}$ from $\btheta_{D}$ to $\hat{\btheta}_{D}$ in $\bar{\boldsymbol{B}}(\btheta_{D},\rho)$. For simplicity, suppress the subscript for $\vtil_{D}$ in below equation. Then
\begin{align*}
    \|\vtil(\hat{\btheta}_{D})-\vtil(\btheta_{D})\|&=\|\vtil(\gamma(1))-\vtil(\gamma(0))\|
  =\|\intt \frac{\dd \vtil(\gamma(t))}{\dd t}\,\dd t\|\\
    &=\Bigl\|\intt \triangledown \vtil(\gamma(t))\cdot (\hat{\btheta}_{D}-\btheta_{D})
    \,\dd t\Bigr\|
    \leq \|\hat{\btheta}_{D}-\btheta_{D}\|\intt \|\triangledown \vtil(\gamma(t))\|\,\dd t 
 <K_{D}\|\hat{\btheta}_{D}-\btheta_{D}\|,
\end{align*}
where $K_{D}:=\sup\{\|\triangledown \vtil(\btheta)\| : \btheta\in\bar{\boldsymbol{B}}(\btheta_{D},\rho)\}$, and the norm are all $l_2$ norms. 
Next we derive {order of the Lipschitz constant}.

Let $f_{D}(v,\btheta_{D})=c_{V\Ub_{\Db}}(\ub_{1:D},v;\btheta_{D})=\exp\{\sum_{j=1}^{D}\log c_{jV}(u_j,v;\btheta_{j})\}$ 
be the density function of 1-factor copula model defined in \eqref{eq-1factor-copula-pdf}, the $j$th element of the {gradient} vector is 
\begin{align}
\label{derr}
  \frac{\partial\vtil(\btheta_{\Db})}{\partial \btheta_j}
  &=\frac{\intt (v\partial f_{D}(v,\btheta_D)/\partial \btheta_j) \dd v
  \cdot \intt f_{D} (v,\btheta_D) \dd v
  -\intt (\partial f_{D}(v,\btheta_D)/\partial \theta_j) \dd v \cdot
 \intt vf_{D}(v,\btheta_D)\dd v}{(\intt f_{D}(v,\btheta_D)\dd v)^2}\notag\\
 &=\bigg\{\left(\intt  f_{D}(v,\btheta_{D})\dd v\right)^{-1}\left[ \intt (v\partial f_{D}(v,\btheta_D)/\partial \btheta_j) \dd v 
  -\vtil_{D}{(\btheta_D)}\times \intt (\partial f_{D}(v,\btheta_D)/\partial \btheta_j) \dd v \right]\bigg\}\\
 &=\left(\intt f_{D}(v,\btheta_{D})\dd v\right)^{-1}
  \left(\intt [v-\vtil_{D}{(\btheta_D)}]
  \times (\partial f_{D}(v,\btheta_D)/\partial \btheta_j) \dd v \right)\notag.
\end{align}
In \eqref{derr}, $\vtil_{D}={\vtil_{D}(\btheta_D)}$ is the proxy variable \eqref{eq-1factor-proxy} defined in 1-factor copula model \eqref{eq-1factor-copula-pdf}. Also since $\partial f_{D}(v,\btheta_D)/\partial \btheta_j
=f_{D}(v,\btheta_D)\left({\partial \log c_{jV}(u_j,v;\btheta_j)}/{\partial \btheta_j}\right)$,
then 
\begin{align*}
\frac{\partial\vtil(\btheta_D)}{\partial \btheta_j}
  &=\frac{\intt (v-\vtil_{D})\,f_{D}(v,\btheta_D)\cdot ({\partial \log c_{jV}(u_j,v;\btheta_j)}/{\partial \btheta_j})\,\dd v}
  {\intt f_{D}(v,\btheta_{D})\dd v}.
\end{align*}
It has the same order as  
\begin{equation}
\label{eqn:der_order}
\frac{\intt (v-\vtil_{D}) \cdot|
 ({\partial \log c_{jV}(u_j,v;\btheta_j)}/{\partial \btheta_j})|\,
  f_{D}(v,\btheta_D) \,\dd v}
  {\intt f_{D}(v,\btheta_{D})\dd v}.  
\end{equation}
In \eqref{eqn:der_order}, let $m_j(v)={\partial \log c_{jV}(u_j,v;\btheta_j)}/{\partial \btheta_j}$, 
$h(v)=\bar{L}_{D}(v)=-D^{-1}\log f_{D}(v,\btheta_{D})$. 
Let $v_D^{*}=\arg\min h(v)$.
Let $t(v)=(v-\vtil_{D})\,|m_{j}(v)|$, $t^{'}(v)
=|m_{j}(v)|+(v-\vtil_{D})\,(\partial |m_{j}(v)|/\partial v)$ 
and $t^{''}(v)=2(\partial|m_{j}(v)|/\partial v)+(v-\vtil_{D})\,(\partial^2 |m_{j}(v)|/\partial v^2)$. 
From equation (2.6) in \citep{tierney1989fully}, equation \eqref{eqn:der_order} becomes
\begin{align*}
 (v_{D}^{*}-\vtil_{D})\,|m_{j}(v_{D}^{*})|+(2D)^{-1}\,
  {[h''(v_{D}^{*})]^{-1}}  \,t''(v_{D}^{*})
-(2D)^{-1}[h''(v_{D}^{*})]^{-2}\,t'(v_{D}^{*})\,h'''(v_{D}^{*})+O(D^{-2}).
\end{align*}
Under the assumptions on the bounded derivatives, 
{together with the proof in Section \ref{sec-proof-proxies}},
${\partial \vtil({\btheta_D})/\partial \theta_j}=O(D^{-1})$. 
Then, the norm of derivatives  $\|\triangledown\vtil(\btheta)\|$
equals to 
  $ \sqrt{\sum_{j=1}^{D}|{\partial\vtil(\btheta_{D})/\partial \btheta_j}|^2}=O(D^{-1/2})$
and $|\vtil(\hat{\btheta}_{D})-\vtil(\btheta_{D})|=O\left(\sqrt{D^{-1}\sum_{j=1}^{D}\|\hat{\btheta}_j-\btheta_j\|_2^2}\right)$.
\end{proof}

\subsection{Proof of Lemma \ref{lipschitz2}}

\begin{proof}
The proof technique is similar to that used in the 1-factor copula case. 
Let $f_{\Db}^{(1)}(\vo,\btheta_{\Db})=c_{\Ub_{\Db},\Vo}(\ub_{\Db},\vo;\btheta_{\Db})$ be the marginal density function defined in Table \ref{tab:notations_bifact}.
Then the {components of the gradient} of the global proxy $\vtil_{0\Db}$ with respect to the parameter vector $\btheta_{\Db}^{(1)}$ {consists of}
\begin{align}
\label{der}
\partial \vtil_{0,\Db}(\btheta_{\Db})/\partial\theta_{jg,0}
&=
\frac{\intt (\vo\partial f_{\Db}^{(1)}(\vo,\btheta_{
\Db})/\partial \theta_{jg,0})\,\dd \vo \cdot
\intt f_{\Db}^{(1)}(\vo,\btheta_{\Db})\,\dd \vo }
{(\intt f_{\Db}^{(1)}(\vo,\btheta_{\Db})\,\dd \vo)^2} \notag\\
&-\frac{{\intt (\partial f_{\Db}^{(1)}(\vo,\bm{\theta}_{\Db})}
/\partial \theta_{jg,0})\,\dd \vo
\intt \vo f_{\Db}^{(1)} (\vo,\btheta_{\Db})\,\dd \vo}
{(\intt f_{\Db}^{(1)}(\vo,\btheta_{D})\,\dd \vo)^2} \notag\\
&=\left(\intt f_{\Db}^{(1)}(v_0,\btheta_{\Db})\,\dd v_0\right)^{-1}
\left(\intt (v_0-\vtil_{0,\Db}{(\btheta_{\Db})})
\times (\partial f_{\Db}^{(1)}(v_0,\bm{\theta}_{\Db})/
\partial \theta_{jg,0})\,\dd v_0\right).
\end{align}
As $\partial f_{\Db}^{(1)}(v_0,\btheta_{\Db})/\partial \theta_{jg,0}=f_{\Db}^{(1)}(v_0,\btheta_{\Db})
\cdot\left({\partial \log c_{\Ub,V_0}(\ub_{\Db},v_0,\btheta_{\Db}) }
/{\partial \theta_{jg,0}}\right)$, 
then \eqref{der} has the same order with
\begin{equation}
\label{der2_tem}
\frac{\intt (v_0-\vtil_{0,\Db}{(\btheta_{\Db})})\cdot 
|{\partial \log c_{\Ub,V_0}(\ub_{\Db},v_0,\btheta_{\Db}) }
/{\partial \theta_{jg,0}}|\, \cdot
f_{\Db}^{(1)}(v_0,\btheta_{\Db}) \, \dd v_0}
{\intt f_{\Db}^{(1)}(v_0,\btheta_{\Db})\,\dd v_0}.
\end{equation}
Let $m_{jg}({v_0})={\partial \log c_{\Ub,V_0}(\ub_{\Db},v_0,\btheta_{\Db}) }
/{\partial \theta_{jg,0}}$, $h(v_0)=\bar{L}_{0\Db}(v_0)
=-D^{-1}\log f_{\Db}^{(1)}(v_0,\btheta_{\Db})$. 
Let ${v_{0,\Db}^{*}}=\arg\min h(v_0)$.  
{With $\vtil_{0,\Db}=\vtil_{0,\Db}(\btheta_{\Db})$},
let $t(v_0)=(v_0-\vtil_{0,\Db})\,|m_{jg}(v_0)|$, $t'(v_0)
=(v_0-\vtil_{0,\Db})\,(\partial|m_{jg}(v_0)|/\partial v_0)+|m_{jg}(v_0)|)$, 
$t''(v_0)=2(\partial|m_{jg}(v_0)|/\partial v_0)+(v_0-\vtil_{0,\Db})\,
(\partial^2|m_{jg}(v_0)|/\partial v_0^2)$.

Then from equation (2.6) in \citep{tierney1989fully}, equation \eqref{der2_tem} becomes
\begin{align*}
 (v_{0,\Db}^{*}-\vtil_{0,\Db})\,|m_{jg}(v_{0,\Db}^{*})|
 +(2D)^{-1}[h''(v_{0,\Db}^{*})]^{-1}\,t''(v_{0,\Db}^{*})
 -(2D)^{-1}[h''(v_{0,\Db}^{*})]^{-2}\,t'(v_{0,\Db}^{*})\,
 h'''(v_{0,\Db}^{*})+O(D^{-2}).
\end{align*}
Under the assumptions on the bounded partial derivatives, 
{together with the proof in Section \ref{sec-proof-proxies}},
$\partial\vtil_{0,\Db}(\btheta_{\Db})/\partial\theta_{jg,0}=O(D^{-1}).$ 

The same logic could be also applies to $\partial \vtil_{0\Db}(\btheta_{\Db})/\partial \theta_{jg}$, $\partial \vtil_{g,\Db}(\btheta_{\Db})/\partial \theta_{jg,0}$ and $\partial \vtil_{g,\Db}(\btheta_{\Db})/\partial \theta_{jg}$ . 
From the above derivation, we conclude that
$$\|\vtil_{0\Db}(\hat{\btheta}_{\Db})-\vtil_{0\Db}(\btheta_{\Db})\|
  =O\left(\sqrt{(2D)^{-1}
 \sum_{g=1}^{G}\sum_{j=1}^{d_g}
  \Bigl\{\|\hat{\btheta}_{jg,0}
  -\btheta_{jg,0}\|_2^2+ \|\hat{\btheta}_{jg}-\btheta_{jg}\|_2^2}\Bigr\}\,\right)$$ 
and 
$$ \|\vtil_{g\Db}(\hat{\btheta}_{D})-\vtil_{g\Db}(\btheta_{\Db})\|
=O\left(\sqrt{(2d_g)^{-1}
\sum_{j=1}^{d_g}
  \Bigl\{\|\hat{\btheta}_{jg,0}
  -\btheta_{jg,0}\|_2^2+ \|\hat{\btheta}_{jg}-\btheta_{jg}\|_2^2}\Bigr\}\,\right),\quad g\in\{1,\ldots,G\}.$$ 
\end{proof}


\medskip


\bibliographystyle{myjmva}
\bibliography{trial}
\end{document}